%% file: pntree.tex
\documentclass[submission,copyright,creativecommons]{eptcs}
\providecommand{\event}{EXPRESS/SOS 2012} 

\usepackage{svn-multi}
\usepackage{cite}
\usepackage{graphicx}
\usepackage{amsmath}
\usepackage{amssymb}
\usepackage{amsthm}
\usepackage{latexsym}
\usepackage{array}
\usepackage{paralist}
\usepackage{color}
\usepackage{textcomp}
\usepackage{mathrsfs}
\usepackage{txfonts}
\usepackage{datetime}
\usepackage{etoolbox}
\usepackage{fancybox}
\usepackage{tikz}
\input xy
\xyoption{all}
\usepackage{ulem}
\usepackage{color}

\newtheorem{example}{Example}
\newtheorem{definition}{Definition}
\newtheorem{theorem}{Theorem}
\newtheorem{lemma}{Lemma}
\newtheorem{corollary}{Corollary}

\svnid{$Id: pntree.tex 486 2012-07-27 18:28:33Z dgr240 $}
\input{macros}

\newif\ifdraft
\draftfalse

\ifdraft
\else
\renewcommand{\remarkDG}[1]{}
\renewcommand{\remarkMDL}[1]{}
\renewcommand{\remarkPRD}[1]{}
\fi

\title{Tree rules in probabilistic transition system specifications with negative and quantitative premises%
  \thanks{Supported by Project ANPCYT PAE-PICT 02272, SeCyT-UNC, 
                 Eramus Mundus Action 2 Lot 13A EU Mobility Programme 2010-2401/001-001-EMA2 and EU 7FP grant agreement 295261 (MEALS).}}
\author{
Matias David Lee$^\dag$  \qquad\qquad
Daniel Gebler$^\ddag$  \qquad\qquad
Pedro R. D'Argenio$^\dag$
\and
\institute{$^\dag$FaMAF -- CONICET \\Universidad Nacional de C\'ordoba\\Ciudad Universitaria, X5000HUA C\'ordoba \\ Argentina}
\email{\{lee,dargenio\}@famaf.unc.edu.ar}
\and
\institute{$^\ddag$Department of Computer Science\\VU University Amsterdam \\ De Boelelaan 1081a, 1081HV Amsterdam \\ The Netherlands}
\email{e.d.gebler@vu.nl}
\ifdraft
\and\institute{\wipinfo}
\fi
}

\def\titlerunning{Tree rules in probabilistic transition system specifications}
\def\authorrunning{M. D. Lee, D. Gebler \& P. R. D'Argenio}

\begin{document}
\maketitle

\begin{abstract}
Probabilistic transition system specifications (PTSSs) in the
\ntmufxnu\ format provide structural operational semantics for
Segala-type systems that exhibit both probabilistic and
nondeterministic behavior and guarantee that bisimilarity is a
congruence.
Similar to the nondeterministic case of the rule format {\it tyft/tyxt}, we
show that the well-foundedness requirement is unnecessary in the
probabilistic setting.
To achieve this, we first define a generalized version of the
\ntmufxnu\ format in which quantitative premises and conclusions
include nested convex combinations of distributions.  
Also this format guarantees that bisimilarity
is a congruence.
Then, for a given (possibly non-well-founded) PTSS in the new format,
we construct an equivalent well-founded PTSS consisting
of only rules of the simpler (well-founded) probabilistic ntree
format.
Furthermore, we develop a proof-theoretic
notion for these PTSSs that coincides with the existing stratification-based meaning
in case the PTSS is stratifiable. This
continues the line of research lifting structural operational semantic
results from the nondeterministic setting to systems with both
probabilistic and nondeterministic behavior.
\end{abstract}

\section{Introduction}

Plotkin's structural operational semantics~\cite{Plotkin81} is a
popular method to provide 
a rigorous interpretation to specification and programming languages. The interpretation 
is given in terms of transition systems.
The method has been formalized with an algebraic flavor as \emph{transition
  systems specifications
  (TSS)}~\cite[etc.]{GrooteVaandrager92,BloomIM95:jacm,Groote93,BolGroote96}.
Basically, a TSS contains a signature, a set of labels, and a set of
rules. The signature defines the terms in the language.  Labels
represent actions performed by a process (i.e., a term over the
signature) in one step of the execution (i.e., one transition).
Rules define how a process should behave (i.e., produce a transition)
in terms of the behavior of its subprocesses. 
That is, rules define
compositionally the transition system associated to each term of the
language.
This technique has been widely studied mainly on the realm of
languages and process algebras describing only non-deterministic
behavior (see~\cite{MousaviEtAl07} for an overview).

The introduction of probabilistic process
algebras~\cite[etc.]{DBLP:journals/iandc/BaetenBS95,DBLP:journals/iandc/GlabbeekSS95}
motivated the need for a theory of structural operational semantics to
define \emph{probabilistic} transition systems.
A few results have appeared in this direction,
notably~\cite{DBLP:journals/entcs/Bartels02,Bartels2004,DBLP:journals/tocl/LanotteT09,klin2008structural,DL-fossacs12}.
All these works introduced rule formats that ensures that bisimulation
equivalence is a congruence for operators whose semantics is defined
within such format.
%
The most general of those formats is the \ntmufxnu\ format ~\cite{DL-fossacs12} that provides semantics in terms of Segala's probabilistic automata~\cite{Segala95}.

\pagebreak[4]

The \ntmufxnu\ format is the probabilistic relative to the
\ntyfxt\ format~\cite{Groote93} extending it in two ways.  First,
it is designed to deal with probabilistic transitions of the form
$t\trans[a]\pi$, where $t$ is a term in the appropriate signature, and
$\pi$ is a distribution on terms. Second, it includes quantitative
premises that allow for probabilistic testing of the form
$\pi(\{t_1,\ldots,t_n\})> q$, that is, it allows to verify if the
probability that the system moves to one state (i.e.\ term) in
$\{t_1,\ldots,t_n\}$ according to $\pi$ is greater than $q\in[0,1]$.
%
The congruence theorem for the
\ntmufxnu\ format~\cite[Thm.~12]{DL-fossacs12} states that if a
probabilistic transition system specification (PTSS) $P$ has all its
rules in \ntmufxnu\ format, then bisimulation equivalence is a
congruence for all operators in $P$.
Unfortunately, \cite{DL-fossacs12}~missed an important condition:
rules have to be well-founded (basically, there should not be a cyclic
dependency on the terms appearing in the premises of the rule).
This paper will correct this mistake.

The well-foundedness condition has also appeared from the very beginning
in the non-deterministic setting.  Most of the formats have it
implicit as they did not allowed lookahead.  Congruence theorems for
formats with lookahead such as
\textit{tyft/tyxt}~\cite{GrooteVaandrager92} or
\ntyfxt~\cite{Groote93} explicitly demanded TSS to be well-founded.
It remained unknown for a while whether such condition was actually
required until Fokkink and van Glabbeek proved it
unnecessary~\cite{FokkinkvanGlabbeek96}.
The proof proceeds by reducing a TSS in \textit{tyft/tyxt}~format (not
necessarily well-founded) to an equivalent TSS containing only so
called \emph{tree rules} (i.e., well-founded rules in \textit{tyft}
format with premises containing only variables instead of
arbitrary open terms).  Similarly, they showed that a TSS in
\ntyfxt\ format can be translated into an equivalent TSS containing
only ntree rules (tree rules with negative premises which are not
necessarily restricted to single variables).

In this paper, we also show that the restriction to well-founded PTSSs
is not necessary to guarantee congruence. We also proceed by
reducing a PTSS in \ntmufxnu\ format to an equivalent PTSS containing
only \textit{pntree rules}.
However, a pntree rule cannot simply be defined as an \ntmufnu\ rule
where positive premises are restricted to the form $x\trans[a]\mu$,
with $x$ and $\mu$ being term and distribution variables, respectively.
It turns out that quantitative premises in \ntmufxnu\ rules are
too limited.  The \ntmufxnu\ format only allows for
quantitative premises of the form $\mu(Y)\gtgeq q$ with $\mu$ being a
distribution variable, $Y$ an infinite set of term variables,
${\gtgeq}\in\{>,\geq\}$, and $q\in[0,1]$.
Instead, the pntree format requires premises of the form
$\theta(Y)\gtgeq q$ where $\theta$ is a nested convex combinations of
products of distribution variables.  We call these objects
\emph{distribution terms}.
%
So, we extend the \ntmufxnu\ format to deal with distribution terms, and prove, more generally, that a
PTSS in the new format ---\,called \ntmufxt\,--- can be translated into an equivalent PTSS
with only pntree rules (hence, well-founded).
Just like for the case of the \ntyfxt\ format, full negative
premises are required, i.e., negative premises in pntree rules cannot
be limited to the form $x\ntrans[a]$, with $x$ being a term
variable.

\medskip

Summarizing, the following results are introduced in this paper:
\begin{itemize}
\item%
  We define the \ntmufxt\ format, which extends the \ntmufxnu\ format
  to deal with distribution terms in quantitative premises.
\item%
  We prove that if a PTSS is in \ntmufxt\ format and it is
  well-founded, then bisimulation equivalence is a congruence for all
  its operators.  This also corrects the mistake in the proof of
  Theorem~12 in~\cite{DL-fossacs12} which omitted to consider the
  well-foundedness hypothesis.
\item%
  We show that for all PTSS in \ntmufxt\ format (not necessarily
  well-founded) there is a PTSS with only pntree rules that defines
  exactly the same probabilistic transition relation (by ``defines'' we
  mean ``has as a \emph{supported model}'')
\item%
  We dropped the well-foundedness hypothesis from the congruence theorem:
  since every pntree rule is also a well-founded \ntmuft\ rule,
  the previous results imply that bisimulation equivalence is a
  congruence for all operators of a (not necessarily
  well-founded) PTSS in \ntmufxt\ format.
\remarkDG{Does statement 4 not states explicitly that `if a PTSS is in \ntmufxt\ format (but not necessarily well-founded), bisimulation equivalence is a congruence for all its operators' - thus make well-foundedness a not necessary precondition? What exactly is in this case the correction?}
\remarkPRD{Actually this is a good point \texttt{:-S}.  Check if the new text is better}
\remarkDG{I believe statements 2 and 4 are still conflicting. Is it not the case that: 1) The proof of bisimilarity is a congruence assumes and works well with the well-foundedness condition. 2) Every PTSS in non well-founded \ntmufxt\  format can be translated to a transition equivalence well-founded PTSS in pntree format 3) Same equiv + pntree are well-founded \ntmufxt\ rules concludes that congruence property is given also for non-well founded PTSS. Do you mean in the beginning of the statement \ntmufxnu\ format which cannot be reduced to a well-founded A\ntmufxnu\ format?}
\remarkPRD{Well, I remove the explanation. There should not be any conflict, now}
\item%
  Besides, in the process, we also redefined important concepts for PTSS
  originally defined for TSS, in particular, the concept of
  ``well supported proof''.
\end{itemize}

\section{Preliminaries}\label{sec:preliminaries}



We assume the presence of an infinite set of (term) variables $\TVar$ and
we let $x, y, z, x',$ $x_0, x_1, \ldots$ range over $\TVar$.
%
A \emph{signature} is a structure $\Sigma = (F, \rank)$, where 
%
\begin{inparaenum}[(i)]
\item%
  $F$ is a set of \emph{function names} disjoint with $\TVar$, and
\item%
  $\rank : F \to \N_0$ is a \emph{rank function} which gives the arity 
  of a function name; if $f \in F$ and $\rank(f)=0$ then $f$ is called
  a \emph{constant name}.
\end{inparaenum}
%
Let $W \subseteq \TVar$ be a set of variables. The set of $\Sigma$-terms 
over $W$, notation $T(\Sigma, W)$ is the least set satisfying: 
%
\begin{inparaenum}[(i)]
\item%
  $W \subseteq T(\Sigma, W)$, and
\item%
  if $f \in F$ and $t_1, \cdots, t_{\rank(f)} \in T(\Sigma, W)$, then
  $f (t_1, \cdots, t_{\rank(f)}) \in T(\Sigma, W)$.
\end{inparaenum}
%
$T(\Sigma, \emptyset)$ is abbreviated as $\closedTerms$; the elements
of $\closedTerms$ are called \emph{closed terms}. $T(\Sigma, \TVar)$ is
abbreviated as $\openTerms$; the elements of $\openTerms$ are called
\emph{open terms}. $\Var(t) \subseteq \TVar$ is the set of variables in
the open term t. 

In order to deal with languages that describe probabilistic behavior we need expressions denoting probability distributions.
Let $\Delta(\closedTerms)$ denote the set of all (discrete) probability
distributions on $\closedTerms$.
We let $\pi, \pi', \pi_0, \pi_1, \ldots$ range over
$\Delta(\closedTerms)$.
As usual, for $\pi \in \Delta(\closedTerms)$ and $T \subseteq
\closedTerms$, we define $\pi(T)=\sum_{t\in T}\pi(t)$.
For $t \in \closedTerms$, let $\delta_t$ denote the Dirac
distribution, i.e., $\delta_t(t)=1$ and  $\delta_t(t')=0$ if $t\not=t'$.
Moreover, the product measure $\prod_{i=1}^n\pi_i$ is defined by
$(\prod_{i=1}^n\pi_i)(t_1,\ldots,t_n)=\prod_{i=1}^n\pi_i(t_i)$. In
particular, if $n=0$, $(\prod_{j\in\emptyset}\pi_j) = \delta_{()}$ is
the distribution that assigns probability 1 to the empty tuple.
Let $g:\closedTerms^n\to\closedTerms$ and recall that 
$g^{-1}(t') = \{ \vec{t} \in\closedTerms^n \mid g(\vec{t})=t'\}$.  Then
$(\prod_{i=1}^n\pi_i)\circ g^{-1}$ is a well defined probability
distribution on closed terms.
In particular, if $g:\closedTerms^0\to\closedTerms$ and $g(())=t$, 
then $(\prod_{j\in\emptyset}\pi_j) \circ g^{-1} = \delta_{()} \circ g^{-1} = \delta_{t}$.

For a term $t\in\openTerms$ we let $\delta_t$ be an
\emph{instantiable Dirac distribution}. That is, $\delta_t$ is a
symbol that takes value $\delta_{t'}$ when variables in $t$
are substituted so that $t$ becomes a closed term $t'\in\closedTerms$.
Let $\DVar = \{\delta_t : t\in \openTerms\}$ be the set of
instantiable Dirac distributions.
%
A \emph{distribution variable} is a variable that takes values on
$\Delta(\closedTerms)$.  Let $\PVar$ be an infinite set of distribution
variables. Let $\mu, \mu', \mu_0, \mu_1, \ldots$ range over $\PVar$
and $\zeta, \zeta', \zeta_0, \zeta_1, \ldots$ range over $\PVar \cup \TVar$.
Let $D \subseteq \PVar$ be a set of distribution variables and $V\subseteq\TVar$ be a set of term variables.  
The set of \emph{distribution terms} over $D$ and $V$, notation $\DT(\Sigma, D, V)$ is the least set satisfying: 
%
\begin{inparaenum}[(i)]
\item%
$D \cup \{\delta_t : t\in T(\Sigma,V)\} \subseteq \DT(\Sigma, D,V)$, and
\item%
  ${\textstyle \sum_{i\in I} p_i (\prod_{n_i \in N_i} \theta_{n_i}) \circ g_i^{-1}} 
  \in \DT(\Sigma, D, V)$ where
    $p_i \in (0,1]$ with $\sum_{i\in I} p_i = 1$,
  each $g_i$ is a
  function s.t.\ $g_i : \closedTerms^{N_i} \rightarrow \closedTerms$,
  and
  $\theta_{n_i}\in \DT(\Sigma, D, V)$.
\end{inparaenum}
 Intuitively, $g_i^{-1}(t)$ decomposes term $t$ into its sub-terms $t_1,\ldots,t_{N_i}$ and probability $\theta(t)$ of term $t$ is calculated as the convex combination of the product probability of its sub-terms $\theta_1(t_i),\ldots,\theta_{N_i}(t_{N_i})$.
%
$\DT(\Sigma, \emptyset, \emptyset)$ is abbreviated as $\closedDTerms$; the elements
of $\closedDTerms$ are actual distributions on terms. $\DT(\Sigma, \PVar, \TVar)$ is
abbreviated as $\openDTerms$.
$\Var(\theta) \subseteq \PVar \cup \TVar$ is the set of  
(distribution and term) variables appearing in $\theta$.
%

A substitution is a mapping that assigns terms to variables.  In our
case we need to extend this notion to distribution terms and
instantiable Dirac distributions.
A \emph{substitution} $\rho$ is a mapping in $(\TVar\cup
\PVar) \to (\openTerms\cup\openDTerms)$ such that $\rho(x)
\in \openTerms$ whenever $x\in \TVar$, and $\rho(\mu) \in
\openDTerms$ whenever $\mu\in \PVar$.
A substitution $\rho$ extends to open terms and sets of terms as usual,
to instantiable Dirac distributions by $\rho(\delta_t)=\delta_{\rho(t)}$ 
and to distribution terms by 
$\rho({\textstyle \sum_{i\in I} p_i (\prod_{n_i \in N_i} \theta_{n_i}) \circ g_i^{-1}}) = 
 {\textstyle \sum_{i\in I} p_i (\prod_{n_i \in N_i} \rho(\theta_{n_i})) \circ g_i^{-1}}
$. Notice that the construction of distribution terms ensures that closed substitution 
instances of distribution terms denote indeed probability distribution.

\section{Probabilistic Transition System Specifications}\label{sec:ptss}

A (probabilistic) transition relation describes the behavior of a
process by prescribing the possible actions it can perform at each
state.
Each action is described with a label on the relation and the
evolution to the next state is given by a probability distribution on
terms.
%
%
We will follow the probabilistic automata style of~\cite{Segala95} which generalize the so
called reactive model~\cite{LarsenSkou91}.  Let $\Sigma$ be a
signature and $\Act$ be a set of labels.  A \emph{transition relation}
is a set ${\trans} \subseteq \PTr$, where $\PTr = \closedTerms \times
\Act \times \Delta(\closedTerms)$.  We denote $(t,a,\pi)\in{\trans}$
by $t\trans[a]\pi$.

Transition relations are usually defined by means of structured
operational semantics in Plotkin's style~\cite{Plotkin81}. 
We follow the approach
of~\cite{GrooteVaandrager92,Groote93,BolGroote96} which provides an
algebraic characterization for transition system specifications.

\begin{definition}\label{def:ptss}%
  A \emph{probabilistic transition system specification} (PTSS) is a
  triple $P = (\Sigma, \Act, R)$ where $\Sigma = (F, r)$ is a signature,
  $\Act$ is a set of labels, and $R$ is a set of rules of the form:
  \[
  \ddedrule{ \{t_k \trans[a_k] \mu_k : k\in K \}\cup
              \{t_l \ntrans[b_l] : l \in L\} \cup
              \{\theta_j (W_j) \gtrless_j q_j : j \in J\} } %
            { t \trans[a] \theta}
  \]
  where
  $K, L, J$ are index sets, 
  $t, t_k, t_l \in \openTerms$, $a, a_k, b_l \in A$, 
  $\mu_k \in \PVar$, $W_j\subseteq \TVar$,
  ${\gtrless_j} \in \{{>},{\geq}, <, \leq \}$, $q_j\in[0,1]$
   and $\theta_j, \theta \in \openDTerms$
\end{definition}

An expression of the form $t \trans[a] \theta$, (resp. $t \ntrans[a]$, $\theta (W)
\gtrless p$) is a \emph{positive literal} (resp. \emph{negative literal, quantitative literal})
where $t \in \openTerms$, $a \in A$, $\theta \in \openDTerms$, $W \subseteq \Var \cup \closedTerms$ and $p \in [0,1]$.
For any rule $r \in R$, literals above the line are called
\emph{premises}, notation $\prem{r}$; the literal below the line is
called \emph{conclusion}, notation $\conc{r}$.
We denote with $\pprem{r}$ ($\nprem{r}$, $\qprem{r}$) the set of
positive (negative, quantitative) literals of the rule $r$.
A rule $r$ is called \emph{positive} if $\nprem{r} = \emptyset$.  A PTSS is called positive if
it has only positive rules. A rule $r$ without premises is called an
\emph{axiom}.
In general, we allow the sets of positive, negative, and quantitative premises to be infinite.

Substitutions provide instances to the rules of a PTSS that, together
with some appropriate machinery, allows us to define probabilistic
transition relations.  Given a substitution $\rho$, it extends to
literals as follows:
%
%
  $\rho(t \ntrans[a]) = \rho(t) \ntrans[a]$, \
  $\rho(\theta(W) \ {\gtrless} \ p) = \rho(\theta)(\rho(W)) \ {\gtrless} \ p$, and
  $\rho(t \trans[a] \theta) = \rho(t) \trans[a] \rho(\theta)$.
Then, the notion of substitution extends to rules as expected.  We say
that $r'$ is a (closed) instance of a rule $r$ if there is a
(closed) substitution $\rho$ so that $r'=\rho(r)$.

 We say that $\rho$ is a \emph{proper substitution of $r$} if for all
 quantitative premises $\rho(\theta(W) \gtrless p)$ of $r$ it holds that
 $\rho(\theta(w)) > 0$ for all $w\in W$.  Thus, if $\rho$ is proper, all
 terms in $\rho(W)$ are in the support of $\rho(\theta)$.  Proper
 substitutions avoid the introduction of spurious terms.  This is of
 particular importance for the conservative extension theorem of
 \cite[Theorem~14]{DL-fossacs12}.
We use only this kind of substitution in the paper. 


As has already been argued many times (e.g.~\cite{Groote93,BolGroote96,vanGlabbeek04}), transition system
specifications with negative premises do not uniquely define a
transition relation and different reasonable techniques may lead to
incomparable models.
In any case, we expect that a transition relation associated to a PTSS
$P$
\begin{inparaenum}[(i)]
\item%
  respects the rules of $P$, that is, whenever the premises of a
  closed instance of a rule of $P$ belong to the transition relation,
  so does its conclusion; and
\item%
  it does not include more transitions than those explicitly
  justified, i.e., a transition is defined only if it is the conclusion of a
  closed rule whose premises are in the transition relation.
\end{inparaenum}
The first notion corresponds to that of model, and the second one to
that of supported transition.

Before formally defining these notions we introduce some notation.
Given a transition relation ${\trans} \subseteq \PTr$, a positive
literal $t \trans[a]\pi$ \emph{holds in} $\trans$, notation
${\trans} \models t \trans[a]\pi$, if $(t, a, \pi) \in {\trans}$.
A negative literal $t \ntrans[a]$  \emph{holds in} $\trans$,
notation ${\trans} \models t \ntrans[a]$,
if there is no $\pi \in \Delta(\closedTerms)$ s.t.\ $(t, a, \pi) \in {\trans}$.
A quantitative literal $\pi(T) \gtrless p$ \emph{holds in} $\trans$,
notation ${\trans} \models \pi(T) \gtrless p$ precisely when $\pi (T)
\gtrless p$.  Notice that the satisfaction of a quantitative literal
does not
depend on the transition relation. We nonetheless
use this last notation as it turns out to be convenient.
Given a set of literals $H$, we write ${\trans} \models H$
if $\forall \phi \in H: {\trans} \models \phi$.

\begin{definition}\label{def:supmodel}%
  Let $P = (\Sigma, A, R)$ be a PTSS. Let ${\trans} \subseteq \PTr$ be
  a probabilistic transition system (PTS). Then ${\trans}$ is \emph{a
   supported model of} $P$ if
  it satisfies that: $\psi \in {\trans}$ iff there is a rule
  $\frac{H}{\chi} \in R$ and a proper substitution $\rho$
  s.t.\ $\rho(\chi) = \psi$ and ${\trans} \models \rho (H)$.
  For ${\trans}$ to be a \emph{model} of $P$ we only require that
   the ``if'' holds, and for ${\trans}$ to be
   \emph{supported by} $P$ we only require that the ``only if'' holds.
\end{definition}

We have already pointed out that PTSSs with negative premises do not
uniquely define a transition relation.  In fact, a PTSS may have more
than one supported model.  For instance, the PTSS with the single
constant $f$, set of labels $\{a,b\}$ and the two rules
$\frac{ f \strans[a] \mu}{ f \strans[a] \delta_{f}}$ and
$\frac{ f \strans[a]\!\!\!\!\!\!\not \quad }{ f \strans[b] \delta_{f}}$,
has two supported models: $\{f\trans[a]\delta_{f}\}$ and
$\{f\trans[b]\delta_{f}\}$.
We will not dwell on this problem which has been studied at length
in~\cite{BolGroote96} and~\cite{vanGlabbeek04} in a non-probabilistic
setting.
Instead we present two different approaches to resolve this problem:
stratification and well supported proofs. 


\subsection{Stratification}\label{sec:stratification}
%
A stratification defines an order on closed positive literals that
ensures that the validity of a transition does
not depend on the negation of the same transition.

\begin{definition}
  \label{def:stratification}%
  Let $P = (\Sigma, A, R)$ be a PTSS.  A function $S: \PTr \to
  \alpha$, where $\alpha$ is an ordinal, is called a
  \emph{stratification} of $P$ (and $P$ is said to be
  \emph{stratified}) if for every rule
  \[r = \ddedrule{\{t_k \trans[a_k] \mu_k : k\in K \}\cup \{t_l \ntrans[b_l]
    : l \in L\} \cup \{\theta_j (W_j) \gtrless q_j : j \in J\} } 
  {t \trans[a] \theta }\]
  and proper substitution
  $\rho : (\TVar\cup\M) \to (\closedTerms\cup\Delta(\closedTerms))$
  it holds that:
  \begin{inparaenum}[(i)]
  \item%
    for all $k \in K$, $S(\rho (t_k \trans[a_k] \mu_k)) \leq
    S(\conc{r})$, and
  \item%
    for all $l \in L$ and $\mu \in \M$, $S(\rho (t_l \trans[b_l] \mu))
    < S(\conc{r}) $.
  \end{inparaenum}
  Each set $S_{\beta} = \{\phi \mid S(\phi)=\beta\}$, with
  $\beta<\alpha$, is called a \emph{stratum}.
  If for all $k \in K$, $S(\rho(t_k\trans[a_k]\mu_k)) < S(\conc{r})$,
  then the stratification is said to be \emph{strict}.
\end{definition}

A transition relation is constructed stratum by stratum in an
increasing manner by transfinite recursion.  
If it has been decided whether a
transition in a stratum $S_{\beta'}$, with $\beta'<\beta$, is valid or
not, we already know the validity of the negative premise occurring in
the premises of a transition $\varphi$ in stratum $S_\beta$ (since all
positive instances of the negative premises are in strictly lesser
strata) and hence we can determine the validity of $\varphi$.
Notice that a stratification does not take quantitative premises into account because 
their satisfaction does not depend on the transition relation.

\begin{definition}\label{def:assoc_with}%
  Let $P = (\Sigma, A, R)$ be a PTSS with a stratification
  $S:\PTr\to\alpha$ for some ordinal $\alpha$.
  For all rules $r$, let $\degPTSS(r)$ be the smallest regular cardinal
  such that $\degPTSS(r) \ge |\pprem{r}|$, 
  and let $\degPTSS(P)$ be the smallest
  regular cardinal such that $\degPTSS(P)\geq\degPTSS(r)$ for all
  $r\in R$.
  The transition relation $\trans_{P, S}$ \emph{associated with} $P$
  (and based on $S$) is defined by
  ${\trans_{P,S}} = \bigcup_{\beta < \alpha} {\trans_{P_\beta}}$,
  where each $\trans_{P_\beta} = \bigcup_{j \leq \degPTSS(P)}
  {\trans_{P_{\beta,j}}}$ and each ${\trans_{P_{\beta,j}}}$ is defined
  by
  \begin{align*}
  {\trans_{P_{\beta,j}}} = \Big\{ \ &
     \psi \mathrel{\big|}
     S(\psi) = \beta \text{ and }
     \exists r \in R \text{ and proper substitution } \rho \mbox{ s.t. }
     \psi = \conc{\rho(r)}, \\[-.5ex]
  & \qquad \textstyle%
    (\bigcup_{ \gamma < \beta} {\trans_{P_\gamma}}) \cup (\bigcup_{ j' < j} {\trans_{P_{\beta,j'}}}) \models
       {\qprem{\rho(r)} \cup \pprem{\rho(r)}}  \text{ and } \\[-.5ex]
  & \qquad \textstyle%
    (\bigcup_{ \gamma < \beta} {\trans_{P_\gamma}}) \models \nprem{\rho(r)} \ \ \Big\}
 \end{align*}
\end{definition}
\noindent
A PTSS $P$ with rules $R=\bigg\{\frac{ f \strans[a] \mu}{ f \strans[a] \delta_{f}}, \frac{ f \strans[a]\!\!\!\!\!\!\not \quad }{ f \strans[b] \delta_{f}}\bigg\}$
can be stratified by $S(f \trans[a] \delta_f) = 0$ and $S(f
\trans[b] \delta_f) = 1$.  This stratification induces the
transition relation ${\trans_{P,S}}=\{f \trans[b] \delta_f\}$. Because 
(non-strict) stratifications allow that positive premises are in the same
stratum as the conclusion, the validity of a premise may depend on a rule with a conclusion 
literal of the same stratum.
In this case, the construction of $\trans_{P_\beta}$ requires
to iterate up to $D(P)$  times, denoted by $\bigcup_{j \leq \degPTSS(P)} {\trans_{P_{\beta,j}}}$,
to decide the the validity of all literals of this stratum.

The  existence of a stratification guarantees the existence of a
supported model. In fact, such model is the one in
Def.~\ref{def:assoc_with} (Theorem~\ref{th:existence:supmodel}). Furthermore,
all stratification define the same supported model (Theorem~\ref{th:weakunicity:supmodel}) which allows to omit the stratification
symbol in ${\trans_{P, S}}$ and use ${\trans_{P}}$ instead. Moreover, 
strict stratification ensures uniqueness of the supported model 
(Theorem~\ref{th:strongunicity:supmodel}). The proofs follow closely
their non-probabilistic counterparts in~\cite{Groote93} (Theorem
2.15, Lemma 2.16 and Theorem 2.18, resp.).
The only actual difference lies on the quantitative premises, which do
not pose any particular problem since their validity depends only on
the substitution.
%

\begin{theorem}\label{th:existence:supmodel}%
	Let $P$ be a PTSS with stratification $S$. Then $\trans_{P, S}$ is a supported model of $P$. 
\end{theorem}

\begin{theorem}\label{th:weakunicity:supmodel}%
  Let $P$ be a PTSS. For all stratifications $S$, $S'$ of $P$ it holds ${\trans_{P, S}} =
  {\trans_{P, S'}}$.
\end{theorem}

\begin{theorem}\label{th:strongunicity:supmodel}%
  Let $P$ be a PTSS with a strict stratification $S$. Then $\trans_{P, S}$ is the only supported model of $P$.
\end{theorem}

\subsection{Proof structures}\label{sec:proofStructure}

 In this section we introduce the notion of \emph{provable rules from a PTSS}.
 To define this notion we use \emph{proof structures} \cite{FokkinkvanGlabbeek96}.
 A proof structure is like a derivation tree where the rules do not share variable names. 
 The connection between the conclusion of a rule $r$ 
 and a premise $\psi$ in other rule is represented by a mapping $\phi$ 
 from rules to literals, i.e. $\phi(r) = \psi$.
 A substitution \emph{matches} with a proof structure if both 
 the conclusion and the premise related by $\phi$ are mapped to the same literal. 
 Thus, matching substitutions translate a proof structure into an actual derivation tree.  As a consequence, a matching substitution applied to a proof structure defines a \emph{provable rule} in which the premises
 are the leaves of the derivation tree and the conclusion is the root.
 The absence of shared variables allows to 
 define substitution on proof structures avoiding name clashes.
 Provable rules will be used in the following way through the paper:  
 given a PTSS $P$ we take the set of provable rules from $P$ 
 with a particular format, these rules will be used to define a 
 a new PTSS $P'$, then we show that $P$ and $P'$ derive the same 
 PTS.

 A PTSS is \emph{small} if for each of its rules the cardinality of its collection of premises does not exceed the cardinality of the set of variables $V$. Small PTSS ensure that there are enough variables to construct the proof structures.

\begin{definition}
 \label{def:proofStructure}
 A \emph{proof structure} is a tuple $\tuple{B, r, \phi}$ such that 
\begin{itemize}
 \item $r \in B$ and $B$ is a set of transition rules which do not have any variables in common,
 \item $\phi$ is an injective mapping from $B\setminus\{r\}$ to the collection of positive premises in $B$, 
          such that each chain $b_0, b_1, \ldots$ in $B$, with $\phi(b_{i+1})$ is a premise of $b_i$,
          is a finite chain. 
\end{itemize}
Let $\rtop(B,r,\phi)$ be the set of all premises of rules in $B$ that are outside the image of $\phi$. Let  $\qtop(B,r,\phi)$ be the set of all quantitative premises in  $\rtop(B,r,\phi)$.
\end{definition}

We introduce a partial well-order $<$ on proof structures to allow inductive reasoning. Define the partial order $<$ by
$(B', r', \phi') < (B,r,\phi)$ iff 
  $B' \subset B$, 
  $\phi'$ is $\phi$ restricted to $B' \setminus \{r'\}$,
  $\rtop(B', r', \phi') \subseteq \rtop(B,r,\phi)$,
  and there is a chain $b_0, b_1, \ldots, b_n$ with 
  $b_0=r$, $b_n=r'$, $n>0$ and $\phi(b_{i+1})$ is a premise of $b_i$.

A substitution $\sigma$ \emph{matches} with the proof structure $(B,r,\phi)$ if $\sigma(\conc{b}) = \sigma(\phi(b))$ for every $b \in B\setminus\{r\}$.

\begin{definition}\label{def:provable}
 Let $H = H_p \cup H_n \cup H_q$ a set of literals s.t. $H_p$, (resp. $H_n$ and $H_q$) is a set of positive (resp. negative and open quantitative) literals.
 A rule $\textstyle{\dedrule{H}{c}}$ is provable from a small PTSS $P = (\Sigma, A, R)$, 
 notation $P \proves \dedrule{H}{c} $, if $c \in H$ or there is a proof structure $(B,r,\phi)$ such that
 each rule in $B$ is in $R$ modulo $\alpha$-conversion and there is a substitution $\sigma$ that matches
 with $(B,r,\phi)$  such that:
\begin{itemize}
 \item $\sigma(\rtop(B, r, \phi) - \qtop(B, r, \phi)) \subseteq H$, 
 \item if $\psi \in \sigma(\qtop(B, r, \phi))$ is a closed quantitative premise then $\psi$ holds,
          otherwise $\psi \in H_q$ and 
 \item $\sigma(\conc{r})=c$.          
\end{itemize}
\end{definition}

 Note that closed quantitative literals do not need to be included in the premise of a provable rule because their validity can be decided without further instantiation. Notice additionally that all negative literals of premises of rules in $B$ are included in $H$ and thus no negative literals can be derived.
%

  \begin{figure}
\begin{minipage}{0.59 \linewidth}
 \centering\small
\begin{tikzpicture}[node distance=1.8cm]
\node (r1) at (0.7,0){$a \trans[a] \delta_{a}$}; 
\node (r2) at (0,-1.5)%
          {$\ddedrule{ \{y_1 \trans[a] \mu_{y_1}\mid y_1 \in Y_1\} \quad s \trans[a] \mu_s \quad \mu_s(Y_1) \geq 1}%
                           {s + t \trans[a] \mu_s}$}; 
\node (r3) at (4.8, -1.5)%
           {$\ddedrule{u \trans[\overline{a}] \mu_u}%
                           {u + v \trans[\overline{a}] \mu_u }$};
\node (r4) at (2.2, -3.5)%
           {$\ddedrule{\{y_2 \trans[b] \mu_{y_2}\mid y_2 \in Y_2\} \quad w \trans[a] \mu_w \quad  x \trans[\overline{a}] \mu_x \quad ((\mu_w \parallel \mu_x) (Y_2) \geq 0.5)}%
                           {w\parallel x \trans[\tau] \mu_w \parallel \mu_x} $};
 \node (r5) at (2, -5.3)%
            {$\ddedrule{ \{y_3 \trans[\textit{ok}] \mu_{y_3} \mid y_3 \in Y_3\} \quad y \parallel z \trans[\tau] \mu_{\parallel} \quad  (\mu_{\parallel} (Y_3) \geq 0.2)}%
               {y\parallel z \trans[\textit{ok}]  \mu_{y'_3}}\;\;{y'_3 \in Y_3}$};
\draw [->] (0.55,-0.3) -- (0.55,-0.8); 
\draw [->] (0.1,-2.1) -- (1.1,-2.9); 
\draw [->] (4.8,-2.1) -- (3.3,-3); 
\draw [->] (2,-4.1) -- (2,-4.7); 
\end{tikzpicture}
\end{minipage}
\hfill
\begin{minipage}{0.33\linewidth}
 \begin{align*}
  \sigma(s) & = a &  \sigma(\mu_s) & = \delta_a\\
  \sigma(y_1) & = a & \text{with } y_1 & \in Y_1 \\
  \sigma(\mu_{y_1}) & = \delta_a & \text{with } y_1 & \in Y_1 \\
  \sigma(w) & = a + t & \sigma(\mu_w) & = \delta_{a}\\
   \sigma(x) & = u + v & \sigma(\mu_x) & = \mu_u \\
   \sigma(y) & = a + t & \sigma(\mu_{\parallel}) & = \delta_a \parallel \mu_u\\
   \sigma(z) & = u + v & & \\
 \end{align*}
\end{minipage}
\caption{An example of proof structure. (See Example~\ref{ex:proofStructure})}
\label{fig:proofStructure}
\end{figure}
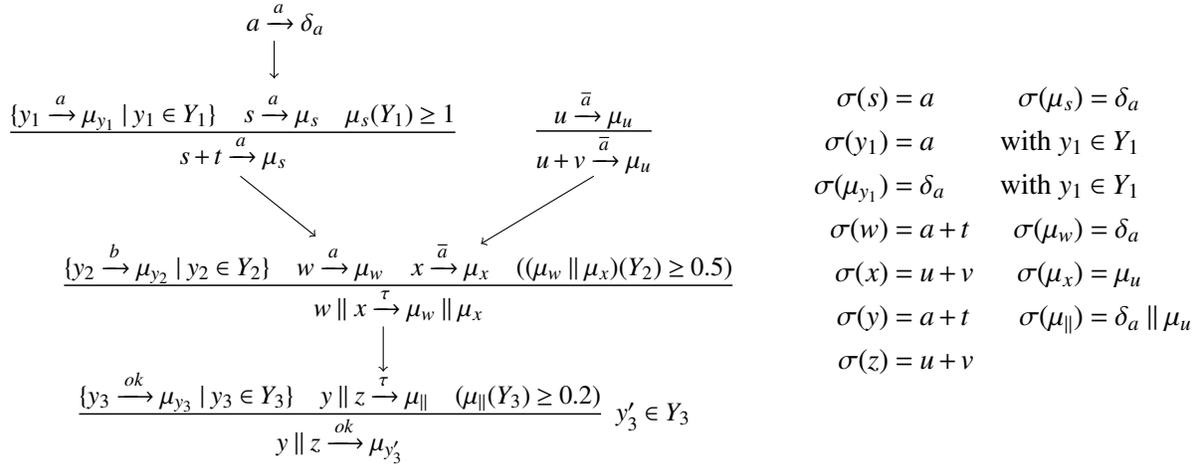

\begin{example} \label{ex:proofStructure}
 Let $P = \tuple{\Sigma, A, R}$ be a PTSS with
 $\{a,{+},{\parallel}\}\subseteq\Sigma$,
 $\{a,\overline{a},b,\tau,ok\}\subseteq A$ and all rules in
 Fig.~\ref{fig:proofStructure} appear in $R$.
 Let $(B,r,\phi)$ the proof structure of
 Figure~\ref{fig:proofStructure} where mapping $\phi$ is represented
 by the arrows.
 Let $\sigma$ be the substitution defined in
 Fig.~\ref{fig:proofStructure}, with $\sigma(\zeta)=\zeta$ for any
 other (term or distribution) variable not specified in the figure.
%
%
 Then the following rule is provable from $P$:
 \begin{equation}\label{eq:ex-provable-rule}
 \ddedrule{u \trans[\overline{a}] \mu_u \quad%
                  \{y_2 \trans[b] \mu_{y_2}\mid y_2 \in Y_2\} \quad ((\delta_a \parallel \mu_u) (Y_2) \geq 0.5) \quad %
                   \{y_3 \trans[\textit{ok}] \mu_{y_3} \mid y_3 \in Y_3\} \quad  ((\delta_a \parallel \mu_u)  (Y_3) \geq 0.2) }
                    {(a+t) \parallel (u + v) \trans[\textit{ok}] \mu_{y'_3} }
 \end{equation}
 Both in Fig.~\ref{fig:proofStructure} and in the above rule we used
 shorthand notations for the different distribution terms.  We write
 $(\mu_w \parallel \mu_x)$ and $(\delta_a \parallel \mu_u)$ instead of
 $(\mu_w \times \mu_x)\circ{\parallel}^{-1}$ and $( (\delta_{()} \circ
 k_a^{-1}) \times \mu_x)\circ{\parallel}^{-1}$, with $k_a(())=a$,
 respectively (trivial summations are omitted).

 Since $\sigma(y_1) = a$ for all $y_1\in y_1$ and $\sigma(\mu_s) =
 \delta_a$, then $\sigma(\mu_s(Y_1) \geq 1) = (\delta_{a}(\{a\}) \geq
 1)$ is closed, and moreover, it holds.  As a consequence, it does not
 appear as a premise of rule~(\ref{eq:ex-provable-rule}).
 Also notice that $\mu_{\parallel}$ was substituted by $(\delta_a
 \parallel \mu_u)$.  This is why we needed to upgrade the format
 of~\cite{DL-fossacs12} to consider the more complex distribution
 terms on the quantitative premises instead of only distribution
 variables.
%
\end{example}

The set of all provable rules from a PTSS can be alternatively defined
in a recursive manner without using the notion of proof structure
(Def.~\ref{def:provClosure}). We prove that both definitions are
equivalent in Lemma~\ref{lemma:closure}.


\begin{definition}\label{def:provClosure}
 The \emph{provable closure} of a PTSS $P = \tuple{\Sigma, A, R}$ 
  is the smallest set $R^\proves$ of rules such that
\begin{itemize}
 \item if $c \in H$ then $\dedrule{H}{c} \in R^\proves$,
 \item if $r \in R$ and there is a substitution $\sigma$ such that 
 \begin{itemize}
    \item for all $p \in \pprem{r} \cup \nprem{r}$ it holds $\dedrule{H}{\sigma(p)} \in R^\proves$  and
    \item for all $p \in \qprem{r}$ if $\sigma(p)$ is not a closed literal then $\dedrule{H}{\sigma(p)} \in R^\proves$, otherwise $\sigma(p)$ holds
 \end{itemize}   
  then  $\dedrule{H}{\sigma(conc(r))} \in R^\proves$.
\end{itemize}
\end{definition}

\begin{lemma}\label{lemma:closure}
 A rule $\dedrule{H}{c}$ is provable from a small PTSS $P = \tuple{\Sigma, A, R}$ iff
 $\dedrule{H}{c} \in R^\proves$.
\end{lemma}

The following lemma is an immediate consequence of Def.~\ref{def:provClosure}.
\begin{lemma}\label{lemma:provability}
 Let $P$ and $P'$ be two PTSS such that all rules in $P'$ are provable from $P$. Then all rules provable from $P'$ are also provable from $P$.
\end{lemma}

\subsection{Well-supported proofs}

In the following we adapt the notion of \emph{well-supported proof}~\cite{vanGlabbeek04} to PTSS.
In the following, we say that literals $t \trans[a] \pi$ and $t \ntrans[a]$ \emph{deny each other}.

\begin{definition}
 \label{def:ws-proof}
 A \emph{well-supported proof} of a closed literal $\psi$ from a PTSS
 $P = (\Sigma, \Act, R)$ is a well-founded, upwardly branching tree of which the nodes
 are labeled by positive or negative literals, such that
\begin{itemize}
 \item the root is labeled by $\psi$, and 
 \item if $\chi$ is the label of the node $q$ and $\{\chi_k \mid k \in K\}$ is the set of labels of 
  the nodes directly above $q$, then:
\begin{itemize}
 \item if $\chi$ is a positive literal then there is a rule $r \in R$ and a closed proper substitution $\rho$ such that 
          $\{\chi_k \mid k \in K\}=\pprem{\rho(r)}\cup\nprem{\rho(r)}$, the quantitative premises $\qprem{\rho(r)}$ are valid and  
          $\conc{\rho(r)} = \chi$,
 \item if $\chi$ is a negative premise then for all $P \vdash \dedrule{N}{\phi}$ with $\phi$ a closed literal denying $\chi$, a 
       literal in $\{\chi_k \mid k \in K\}$ denies a literal in $N$.
\end{itemize}
\end{itemize}
A literal $\psi$ is \emph{ws-provable}, notation $P \wsproves \psi$, if there is a well-supported proof of
$\psi$ from $P$. A literal $\psi$ is \emph{ws-refutable} if there is a literal $\psi'$ ws-provable from $P$ and $\psi$ denies $\psi'$.
\end{definition}

Notice that nodes in the proof tree of Def.~\ref{def:ws-proof} are not
quantitative literals.  This is due to the fact that the validity of
closed quantitative literals is already known.  In fact, the
definition requires that all quantitative literal introduced by a rule
$r$ should become valid after substitution.
%

We say that a PTSS $P$ is \emph{complete} if for all closed literal $t \ntrans[a]$, $P \vdash_{ws} t\trans[a] \pi$ for some distribution $\pi$ or $P \vdash_{ws} t\ntrans[a]$.
In addition, $P$ is \emph{consistent} if there are no pair of literals
derived from $p$ that deny each other.
We will focus only on complete PTSSs. 
\remarkDG{Main theorem does not require completeness, so for which part and why is completeness required?}
\remarkPRD{You are right and I thought of it, but since we are defining in this paper the concept of a model based on well-supported proof maybe it is ok to keep it.  Anyway I am not really aiming to push this.  If space is needed, this would be a candidate to remove}
The transition relation based on well-supported proofs associated to 
a (complete) PTSS $P$ (denoted by ${\trans_{ws}}$) is the set of ws-provable transitions
of $P$.

\begin{lemma}\label{lemma:completeThenCons}
 Let $P$ be a PTSS. If $P$ is complete then it is also consistent.
\end{lemma}

 Lemma~\ref{lemma:completeThenCons} allows us to show that, for any stratifiable PTSS, the model obtained using well-supported proofs coincides with the model obtained through stratification. 
 Notice that this does not imply that the methods are equivalent: it could be the case
 that a PTSS is complete but not stratifiable (see \cite[Prop. 27]{vanGlabbeek04}).
 
\begin{lemma}\label{lemma:WSPsubsumesStra}
 Let $P$ be a PTSS with stratification $S$ and $\psi$ a positive or negative literal, 
 then  $\psi \in {\trans_{ws}}$ iff ${\trans_{P, S}} \models \psi$.
\end{lemma}

The proof of this lemma follows the same structure of its non-probabilistic counterpart (see \cite[Prop. 25]{vanGlabbeek04}).

The next lemma states that it suffices to show that the same rules
having only negative premises are provable in two different PTSSs to
state that these PTSSs define the same set of ws-provable transitions.
%
%

\begin{lemma}\label{lemma:sameNegRules}
 Let $P$ and $P'$ be two PTSSs over the same signature such that
 $P \proves \dedrule{H}{c}$ iff $P' \proves \dedrule{H}{c}$ for all closed rule $\dedrule{H}{c}$ with $H$ containing only negative premises.  Then  
 $P \proves_{ws} \psi$ iff $P' \proves_{ws} \psi$ for all closed literal $\psi$. 
\end{lemma}

\section{The \bntmufxt\  format}

In this section we revise the \ntmufxnu\ format of~\cite{DL-fossacs12}
adapting it to the richer quantitative premises introduced before.
Furthermore we correct some mistakes of~\cite{DL-fossacs12}.

Before, we recall the notion of bisimulation on PTSs~\cite{LarsenSkou91}. 
Given a relation ${\relR} \subseteq \closedTerms \times \closedTerms$,
a set $Q \subseteq \closedTerms$ is $\closed{\relR}$ if for all $t \in
Q$ and $t' \in \closedTerms$, $t \relR t'$ implies $t' \in Q$
(i.e. ${\relR}(Q) \subseteq Q$).
If a set $Q$ is $\closed{\relR}$ we write $\closed{\relR}(Q)$.
It is easy to verify that if two relation ${\relR},
{\relR'} \subseteq \closedTerms \times \closedTerms$ are such that
${\relR'}\subseteq{\relR}$, then for all set $Q \subseteq
\closedTerms$, $\closed{\relR}(Q)$ implies $\closed{\relR'}(Q)$.

\begin{definition}
  A relation ${\relR} \subseteq \closedTerms \times \closedTerms$ is a
  bisimulation if ${\relR}$ is symmetric and for all $t, t' \in
  \closedTerms$, $\pi \in \Delta(\closedTerms)$, $a\in A$,
  \begin{quote}
    $t \relR t'$ and $t \trans[a] \pi$ imply that there exists
    $\pi' \in \Delta(\closedTerms)$ s.t.\ $t' \trans[a] \pi'$ and
    $\pi \relR \pi'$,
  \end{quote}
  where $\pi \relR \pi'$ if and only if $\forall Q \subseteq
  \closedTerms: \closed{\relR}(Q) \limp \pi(Q) = \pi'(Q)$.
  We define bisimilarity $\bisim$ as the smallest relation that
  includes all other bisimulations.
  It is well-known that $\bisim$ is itself a bisimulation and an
  equivalence relation.
\end{definition}

Let $\{Y_l\}_{l\in L}$ be a family of sets of term variables with the same cardinality.
The $l$-th element of a tuple $\vec{y}$  is denoted by $\vec{y}(l)$.  For a set of
tuples $T=\{\vec{y_i} \mid i\in I\}$ we denote the $l$-th projection by 
$\pi_l(T)=\{\vec{y_i}(l) \mid i\in I\}$.
Fix a set $\diag{Y_l}_{l\in L} \subseteq \prod_{l\in L} Y_l $ such that:
%
\begin{enumerate}[(i)]
\item%
  for all $l \in L$, $\pi_l(\diag{Y_l}_{l\in L}) = Y_l$; and
\item%
  for all $\vec{y}, \vec{y'} \in \diag{Y_l}_{l\in L}$, $(\exists l \in L :
  \vec{y}(l) = \vec{y'}(l)) \limp \vec{y} = \vec{y'}$.
\end{enumerate}
Property (ii) ensures that different $\vec{y}, \vec{y'} \in \diag{Y_l}_{l\in L}$ differ in all positions and by property (i) every variable of every $Y_l$ is used in one $\vec{y} \in \diag{Y_l}_{l\in L}$.
$\Diag$ stands for ``diagonal'', following the intuition that each
$\vec{y}$ represents a coordinate in the space $\prod_{l\in L} Y_l$,
then $\diag{Y_l}_{l\in L}$ can be seen as the line that traverses
the main diagonal of the space.
Notice that, letting $L$ be a natural number,  for $Y_l=\{y_l^0,y_l^1,y_l^2,\ldots\}$ a possible
definition for $\diag{Y_l}_{l\in L}$ is $\diag{Y_l}_{l\in L} =
\{(y_0^0,y_1^0,\ldots,y_L^0),(y_0^1,y_1^1,\ldots,y_L^1),(y_0^2,y_1^2,\ldots,y_L^2),\ldots\}$.

\begin{definition} \label{def:ntmufxnu}
  Let $P = (\Sigma, A, R)$ be a PTSS.
  A rule $r\in R$ is in \emph{\ntmuft\ format} 
  if it has the following form %
  \[
 \ddedrule%
      {
	{\textstyle
         \bigcup_{m\in M}
           \{ t_m(\vec{z})\trans[a_m]\mu_m^{\vec{z}} : \vec{z}\in \mathcal{Z} \} \ \cup \
         \bigcup_{n\in N}
           \{ t_n(\vec{z})\ntrans[b_n] : \vec{z}\in \mathcal{Z} \} \  \cup \
	  }
         {\textstyle
		\{ \theta_l(Y_l)\gtgeq_{l,k} p_{l,k} : l\in L, k\in K_l \}
	  }
      }
      {f(x_1,\ldots,x_{\rank(f)}) \trans[a] \theta }
  \]
  %
%
  with ${\gtgeq_{l,k}} \in \{>, \geq\}$
  for all $l\in L$ and $k\in K_l$, and it satisfies the following conditions:
  \begin{enumerate}
  \item\label{item:conditions_on_cardinality}%
    Each set $Y_l$ should be at least countably infinite, for all
    $l\in L$, and the cardinality of $L$ should be strictly smaller
    than that of the $Y_l$'s.
  \item\label{item:condition_on_Z}%
    $\mathcal{Z} = \diag{Y_l}_{l\in L} \times\prod_{w\in W}\{w\}$, with 
    $W\subseteq \V \backslash \bigcup_{l\in L} Y_l$.
  \item\label{item:condition_mus_are_different}%
    All variables $\mu_m^{\vec{z}}$, with $m\in M$ and
    $\vec{z}\in \mathcal{Z}$, are different.
    %
\remarkPRD{Notice that there is one item less thanks to the removal of the $\zeta$'s in the $\theta$'s and a simplification of items~\ref{item:conditions_on_nonrepeating_z}, \ref{item:conditions_on_variables_of_theta}, and \ref{item:conditions_on_distribution_terms}.}
  \item\label{item:conditions_on_nonrepeating_z}%
    For all $\vec{z}, \vec{z'} \in \mathcal{Z}$, $m\in M$, if 
   $\mu^{\vec{z}}_m, \mu^{\vec{z'}}_{m} \in \Var(\theta) \cup (\cup_{l \in L} \Var(\theta_l))$ then $\vec{z}= \vec{z'}$.
  \item\label{item:conditions_on_Y}%
    For all $l\in L$,
    $Y_l \cap \{x_1,\ldots,x_{\rank(f)}\} = \emptyset$, %
    and %
    $Y_l \cap Y_{l'} = \emptyset$ for all $l'\in L$, $l\neq l'$.
  \item\label{item:conditions_on_z}%
    All variables $x_1,\ldots,x_{\rank(f)}$ are different.
  \item\label{item:conditions_on_variables_of_theta}%
    For all $l\in L$,
    $\Var(\theta_l) \cap (\{x_1,\ldots,x_{\rank(f)}\} \cup \bigcup_{l'\in L} Y_{l'}) = \emptyset$.
    %
  \item\label{item:conditions_on_terms}%
    $f\in F$ and for all $m\in M$ and $n\in N$, $t_m, t_n \in
    \openTerms$.  In all cases, if $t\in\openTerms$ and $\Var(t)
    \subseteq \{w_1,\ldots,w_H\}$, $t(w'_1,\ldots,w'_H)$ is the same
    term as $t$ where each occurrence of variable $w_h$ (if it appears
    in $t$) has been replaced by variable $w'_h$, for $1 \leq h \leq  H$.
  \item\label{item:conditions_on_distribution_terms}%
    $\theta, \theta_l \in \openDTerms$ for all $l \in L$.
  \end{enumerate}
  A rule $r\in R$ is in \emph{\ntmuxt{} format} if its form is like above
  but has a conclusion of the form
 $x \trans[a] \theta$
  and, in addition,
  it satisfies the same conditions as above only that whenever we write $\{x_1,\ldots,x_{\rank(f)}\}$, we should write $\{x\}$.
%
  A rule $r\in R$ is in \emph{\nxmuft{} format} if it is in 
  \ntmuft{} format and the sources of its positive premises are term variables.
  $P$ is in \ntmuft{} (resp.\ \emph{\ntmuxt{}},\ \emph{\nxmuft{}}) \emph{format}
  if all its rules are in \ntmuft{} (resp.\ \ntmuxt{},\ \nxmuft{}) format.
  $P$ is in \emph{\ntmufxt{} format} if each of its rules is either in
  \ntmuft{} format or \ntmuxt{} format.
\end{definition}
The rationale behind each of the restrictions are discussed
in~\cite{DL-fossacs12} in depth.
In the following we briefly summarize it.
Term variables $x_1,\ldots,x_{\rank(f)}$ appearing in the source of the
conclusion are binding.  Variables in $\bigcup_{l\in L} Y_l$ and those
appearing in instantiable Dirac distributions are also binding when
appearing in quantitative premises.  Therefore they need to be all
different.  This is stated in conditions
\ref{item:condition_mus_are_different}, \ref{item:conditions_on_Y},
and \ref{item:conditions_on_variables_of_theta}.
Distribution variables in 
$\{\mu_m^{\vec{z}}\mid {m\in M} \land {\vec{z}\in\mathcal{Z}}\}$
are also binding when appearing on the target of a positive
premise. Hence they also need to be different, which is stated in
condition \ref{item:conditions_on_z}.
If $Y_l$ is finite, quantitative premises will allow to count the
minimum number of terms that gather certain probabilities.  This goes
against the spirit of bisimulation that measures equivalence classes of
terms regardless of the size of them.  Therefore $Y_l$ needs to be
infinite (condition~\ref{item:conditions_on_cardinality}).
Condition~\ref{item:conditions_on_nonrepeating_z} is more subtle;
together with each set of premises
$\{t_m(\vec{z})\trans[a_m]\mu_m^{\vec{z}} : \vec{z}\in \mathcal{Z} \}$
it ensures a symmetric behaviour of terms $t_m(\vec{z})$ for every
possible instantiation of variables $\vec{z}$.  A clear example that
shows the need for this symmetry is provided in~\cite{DL-fossacs12}.
The need for the source of the conclusion and targets of positive
premises to have a particular shape is the same as in the
\emph{tyft/tyxt} format~\cite{GrooteVaandrager92}.
Conditions~\ref{item:condition_on_Z},
\ref{item:conditions_on_terms},
and~\ref{item:conditions_on_distribution_terms} are actually notations
and definitions.

%

The definition provided here corrects some mistakes inadvertently
introduced in the \ntmufxnu\ format in~\cite{DL-fossacs12}, more
precisely on the quantitative premises and condition~4 in Def.~11
(which corresponds to our
condition~\ref{item:conditions_on_nonrepeating_z}).
Another mistake in~\cite{DL-fossacs12} was omitting to require that
PTSS are well-founded as hypothesis for the congruence theorem.  This
is corrected in the following, where we extend the congruence theorem
to the \ntmufxt\ format.
\footnote{Both issues are explained in detail in the corrigendum of \cite{DL-fossacs12}:
http://cs.famaf.unc.edu.ar/~lee/publications/corrigendum-Fossacs2012.pdf}


\begin{definition}
 Let $W$ be a set of positive and quantitative premises.
 The \emph{dependency directed graph} of $W$ is given by $G_W = (V, E)$ with
 $V = \textstyle{\cup_{\psi \in W} \Var(\psi)}$ and
$E = \{\tuple{x, \mu} \mid t \trans[a] \mu, x \in \Var(t)\} \cup \{\tuple{\zeta, y} \mid (\theta(Y)\gtgeq p) \in W, \zeta \in \Var(\theta), y \in Y\}$.
%
 We say that $W$ is \emph{well-founded} if any backward chain of edges in  $G_W$ is finite.
%
%
 Define for each $x \in V$, $\nvdg(x) = \sup(\{\nvdg(y) + 1 \mid (y,x) \in E \})$,
 where $\sup(\emptyset) = 0$.   
 A rule is called \emph{well-founded} if its set of positive and quantitative premises is well-founded. 
 A PTSS is called \emph{well-founded} if all its rules are well-founded. 
\end{definition}


\begin{theorem}\label{th:congruence}
 Let $P$ be a well-founded stratifiable PTSS in $\ntmufxt$\ format. Then $\bisim$ is a congruence relation for all operators defined in $P$.
\end{theorem}

\section{\bntmufxt\ format reduces to pntree}\label{sec:reduction_pntree}

The reduction procedure requires results from unification theory over infinite domains.  
Instead using the result presented in \cite{Fokkink1997183}, we use the 
variation presented in \cite[Lemma 3.2]{FokkinkvanGlabbeek96} that proves some extra
properties needed to prove our main result.

\begin{definition}
 A substitution $\sigma$ is a \emph{unifier} for a substitution $\rho$ if $\sigma \rho = \sigma$.
 In this case, we say that $\rho$ is \emph{unifiable}.
\end{definition}

\begin{lemma}\label{lemma:unification}
 If a substitution $\rho$ is unifiable, then there is a unifier $\hat{\sigma}$ for $\rho$  such that:
 \begin{inparaenum}[(i)]
  \item each unifier $\sigma$ for $\rho$ is also a unifier for $\hat{\sigma}$
  \item if $\rho(\zeta) = \zeta$ then $\hat{\sigma}(\zeta) = \zeta$, for all $\zeta\in\TVar\cup\PVar$, and
  \item if $\rho^n(\zeta)$ is a variable for all $n \geq 0$ then $\hat{\sigma}(\zeta)$ is a variable.
 \end{inparaenum}
We call $\hat{\sigma}$ the most general unifier.
\end{lemma}

The main theorem~\ref{th:ntumufxtPNTree} showing that every PTSS in \ntmufxt-format 
can be reduced to a transition equivalent PTSS in pntree format is 
developed incrementally. 
First of all, we show that every \ntmuxt-rule can be expressed by a set of \ntmuft-rules
by replacing the source variable of the conclusion with an appropriate context 
$f(\vec{x})$ (Lemma~\ref{lemma:ntmuft}).
Secondly, we show that for all PTSS $P$ in \ntmuft\ format 
there is a PTSS $P'$ in \nxmuft\ format such that 
$P \proves \dedrule{H}{c}$ iff $P' \proves \dedrule{H}{c}$
for all rules $\dedrule{H}{c}$ in \nxmuft\ format (Lemma~\ref{lemma:nxmuft}).
Notice that this result implies that 
$P \proves \dedrule{H'}{c}$ iff $P' \proves \dedrule{H'}{c}$
for all rule $\dedrule{H'}{c}$ with $H'$ a set of closed negative premises, then
by Lemma~\ref{lemma:sameNegRules}, $P$ and $P'$ are equivalent. 
Finally, we prove that for all PTSS $P$ in \nxmuft\ format 
there is a PTSS $P'$ in  pntree format
(a PTSS in well-founded \nxmuft\ format without free variables),
such that for every closed transition rule $\dedrule{H}{c}$  
with only negative premises,
$P \proves \dedrule{H}{c}$ iff $P' \proves \dedrule{H}{c}$ (Lemma~\ref{lemma:pntree}).
Again, by Lemma~\ref{lemma:sameNegRules}, $P$ and $P'$ are
equivalent. 
This series of lemmas leads to the main theorem stating that every PTSS consisting of rules in the \ntmufxt\ format can be reduced to a transition equivalent PTSS in the more restrictive pntree format. Furthermore, this shows also that the rules of a PTSS in \ntmufxt\ format do not have to be well-founded in order to guarantee that the bisimilarity of the induced PTS is a congruence.

The reduction of proof structures follows the logic of \cite{FokkinkvanGlabbeek96}. In the probabilistic setting we need to treat additionally quantitative premises as follows: While substitutions replace distribution variables by distribution terms the substitution $\rho(\theta(Y) > p)$ leads to a well-defined quantitative literal ($\rho$ is defined as $\rho(y) = y $ for all $y \in Y$). Because by construction $\sigma$ unifies $\rho$ we have that whenever $\sigma(\theta(Y) > p)$ then also $\sigma(\rho(\theta(Y) > p))$. This shows the satisfaction of the quantitative premises.

\begin{lemma}\label{lemma:ntmuft}
  Let $P = (\Sigma,A, R)$ be a stratifiable PTSS in \ntmufxt\ format.
  Then there is a stratifiable PTSS $P' = (\Sigma, A, R')$ in
  \ntmuft\ format that is transition equivalent to $P$.
\end{lemma}

\begin{lemma}\label{lemma:nxmuft}
  Let $P = (\Sigma,A, R)$ be a PTSS in \ntmuft\ format.
  Then there is a PTSS $P' = (\Sigma, A, R')$ in
  \nxmuft\ format such that 
  $P \proves \dedrule{H}{c}$ iff $P' \proves \dedrule{H}{c}$
  for all rule $\dedrule{H}{c}$ in \nxmutt\ format.
  (A rule is in \nxmutt\ format if the source of every positive
  premise is a term variable and its target is a distribution
  variable.)
\end{lemma}

\begin{proof}
 Define $P' = (\Sigma, A, R')$ such that $r \in R'$ iff 
 $r$ is a provable rule from $P$ in \nxmuft\ format.
 The right to left implication follows straightforward from Lemma~\ref{lemma:provability}.

 For the left to right implication we proceed by induction on the partial order over proof structures.
 Suppose $P \proves \dedrule{H}{c}$, with a rule $\dedrule{H}{c}$ in \nxmutt\ format, 
 and let $(B, r, \phi)$ be a proof structure for  $\dedrule{H}{c}$ over $P$.
 Then by Def.~\ref{def:provable} there is substitution $\sigma$ s.t.
 \begin{inparaenum}[(a)]
 \item%
 $\sigma(\rtop(B, r, \phi) - \qtop(B, r, \phi))  \subseteq H$,
 \item%
 closed quantitative premise in $\sigma(\qtop(B, r, \phi))$ hold,
 \item%
 open quantitative premise in $\sigma(\qtop(B, r, \phi))$ belong to $H$,
 and
 \item%
 $\sigma(\conc{r}) = c$.  
 \end{inparaenum}

 From $(B, r, \phi)$ we construct recursively a substructure $(B', r, \phi')$ 
 which is a proof structure for a rule $r' \in R'$, i.e. $r'$ is in \nxmuft\  format, 
 such that  $\sigma(\conc{r'}) = c$ and 
 for each premise $c'$ of $\sigma(r')$ the rule $\dedrule{H}{c'}$ is provable from $R'$
  i.e.  $\dedrule{H}{c} \in R'^\proves$ or $c'$ is a valid closed quantitative literal.
 Then, by Lemma~\ref{lemma:closure}, $\dedrule{H}{c}$ is provable from $P'$.
  Furthermore, we construct a partial substitution $\rho$ which is unified by $\sigma$, 
  i.e. if $\rho(x)$ is defined then $\sigma(\rho(x)) = \sigma(x)$.
  In this construction $\rho^0$ is defined as the identity function.
 We proceed with the definitions of the transition rules $B'$ and the substitution $\rho$: 
\begin{enumerate}[(i)]
  \item $r \in B'$. 
  \item \label{cond:B&rho:ii}
  If $b \in B \setminus \{r\}$ and $\phi(b)$ is a premise $t_m(\vec{z}) \trans[a_m] \mu_{m}^{\vec{z}}$
  of a rule in $B'$ s.t there is $k \geq 0$ with:
 \begin{enumerate}
  \item $\rho^i(t_m(\vec{z}))$ is defined for $i = 0, \dots, k$ 
  \item $\rho^i(t_m(\vec{z}))$ are variables for $i = 0, \dots, k-1$ 
  \item $\rho^k(t_m(\vec{z}))$ has the form $f(t_1, \dots, t_{\rank(f)})$ with $t_i \in \openTerms$
 \end{enumerate}
 then $b \in B'$. 
 Notice that the conditions can be satisfied only if  
 $\rho^i(t_m(\vec{z}))$ is a variable for $i = 0, \dots, k-1$. 
 Moreover $\rho^0(t_m(\vec{z})) = t_m(\vec{z})$ is a variable. 
 In addition, this variable belongs to  $\vec{z}$.

\item \label{cond:B&rho:iii}%
 Since $\sigma$ matches with $(B, r, \phi)$, 
 $\sigma(\conc{b}) = \sigma(t_m(\vec{z}) \trans[a_m] \mu_{m}^{\vec{z}})$. 
 Because the rule format restricts the form of the conclusion $\conc{b}$, 
 then we can rewrite the last equality by:
  $
  \sigma(f(x_1,\ldots,x_{\rank(f)}) \trans[a] \theta
   ) =
  \sigma(t_m(\vec{z}) \trans[a_m] \mu_{m}^{\vec{z}})
  $
 In addition, $\sigma$ unifies the partial substitution $\rho$,
 then if $\rho^{k-1} (t_m(\vec{z}))$ is a variable it holds:
 $\sigma(t_m(\vec{z})) = \sigma \rho^k (t_m(\vec{z})) = \sigma(f(t_1, \dots, t_n)). $

 Because $\conc{b}$ has the form $f(x_1,\ldots,x_{\rank(f)}) \trans[a] \theta$ 
 it holds $\sigma(x_j) = \sigma(t_j)$ for $j = 1, \dots, \rank(f)$ and 
 $\sigma(\theta) =  \sigma(\mu_m^{\vec{z}})$.
 Define $\rho(x_j) = t_j $ for $j = 1, \dots, \rank(f)$
 (here we define the left side of a conclusion of a rule in $B'\setminus r$).   
 Besides, define $\rho(\mu_m^{\vec{z}}) = \theta$.
 Notice that this extension of $\rho$ is unified by $\sigma$ and, 
 by Def.~\ref{def:provable}, the variables $x_j$ and $\mu_m^{\vec{z}}$ 
 appear only in this rule, then we are not redefining substitution $\rho$.
\item Define $\rho(\zeta) = \zeta$ for all variable $\zeta$
 if $\zeta$ is not defined for $\rho$. 
 Substitution $\sigma$ unifies this extension of $\rho$.
\item Finally, $\phi'$ is the restriction of $\phi$ to $B'\setminus\{r\}$.  
(Notice that the substitution $\rho$ is defined for the the right side 
of a positive premise in the image of $\phi'$ in item (\ref{cond:B&rho:iii}).)

\end{enumerate}

Substitution $\sigma$ unifies substitution $\rho$, by Lemma~\ref{lemma:unification},
there is a substitution $\rho'$ which unifies $\rho$ and:
\begin{enumerate}[($\rho'$i)]
 \item \label{unification:i} 
 $\sigma \rho' = \sigma$.
 \item \label{unification:ii}
 If $\rho(\zeta) = \zeta$ then $\rho'(\zeta) = \zeta$, with  $\zeta$ a term or distribution variable.
 \item \label{unification:iii} 
 If $\rho^k(\zeta)$ is a variable for $k \geq 0$ then $\rho'(\zeta)$ is a variable. 
\end{enumerate}
 
The proof structure  $(B', r, \phi')$ and the substitution $\rho'$ are completely defined,
now we can prove that $\rho'$ matches with $(B', r, \phi')$. 
Let $b$ a rule used to construct $B'$ and consider the substitution $\rho$.
Recall that the conclusion of $b$ has the form
$f(x_1,\ldots,x_{\rank(f)}) \trans[a] \theta$ and
$\phi'(b) = t_m(\vec{z}) \trans[a_m] \mu_{m}^{\vec{z}}$ is such that 
$\rho^k (t_m(\vec{z})) = f(t_1, \dots, t_{\rank(f)}) =\rho(f(x_1, \dots, x_{\rank(f)}))$ 
by (\ref{cond:B&rho:ii}) and the definition of $\rho$ for $x_i$ in 
(\ref{cond:B&rho:iii}). 
Since $\rho'$ unifies $\rho$ then 
\[\begin{array}{rcccccl}
  	\rho'(\phi'(b))
        & = & \rho'( t_m(\vec{z}) \trans[a_m] \mu_{m}^{\vec{z}})
   	& = &\rho'( \rho^k( t_m(\vec{z})) \trans[a_m] \rho(\mu_{m}^{\vec{z}})))
        & = & \\
   	& = &\rho'( \rho(f(x_1, \dots, x_{\rank(f)}) \trans[a_m]  \theta )
   	& = &\rho'( f(x_1, \dots, x_{\rank(f)}) \trans[a_m] \theta)
	& = &\rho'(\conc{b})\\
  \end{array}
\]
Then the substitution $\rho'$ matches with the proof structure $(B', r , \phi')$.

To show that the rule 
$s = \rho' \left(\ddedrule{\{h : h \in \rtop(B', r , \phi'), h \text{ is not a closed quantitative premise}\}}{\conc{r}}\right)$
is provable (Def.~\ref{def:provable}), it remains to show that if a quantitative premise 
in $\qtop(B', r , \phi')$ is closed then it is also valid.
 Let $\psi \in \qtop(B', r , \phi')$ be a quantitative premise. 
 Then if  $\rho'(\psi)$ is closed, since $\sigma$ unifies $\rho'$, it holds that
 $\sigma(\psi) = \sigma(\rho'(\psi)) = \rho'(\psi)$,
 which implies that also $\sigma(\psi)$ is a closed literal. 
 Because the rule 
 $\sigma \left(\ddedrule{\{h : h \in \rtop(B, r , \phi), h \text{ is not a closed quantitative premise}\}}{\conc{r}} \right)$
 is provable we have that $\sigma(\psi)$ holds and therefore also $\rho'(\psi)$ holds.

Finally we prove that the rule $s$ is in \nxmuft\ format. 
From the construction by $\rho$ we know that if $x$ is
s.t. $\rho(x) \neq x$ then $x$ satisfies one of the following conditions:

\begin{enumerate}
 \item $x$ appears in the left-hand side of a conclusion of a rule in $B'\setminus \{r\}$,
 \item $x$ appears in the right-hand side of a positive premise in the image of $\phi'$. 
\end{enumerate}

 Then if $g(x_1, \dots, x_m) \trans[b] \theta$ is the conclusion of $r$, 
 $\rho(x_j) = x_j$ for $j=1, \dots, m$ and, 
 hence $\rho'(x_j) = x_j$ because of ($\rho'$\ref{unification:ii}).
 On the other hand, if $\zeta \in \Var(\theta)$ is a variable that appears in the right-hand 
 side of a positive premise in the image of $\phi'$, i.e. $\zeta$ is a distribution variable, 
 we have $\rho'(\zeta) \in \openDTerms$ and then 
 $\rho'(\zeta) \in \openDTerms$. 
 Therefore the conclusion $\rho'(g(x_1, \dots, x_m) \trans[b] \theta)$ of $s$ 
 has the form $g(x_1, \dots, x_m) \trans[a] \rho'(\theta)$ as the \nxmuft\ format demands. 

 We continue with the premises of $s$.   
 Let $\rho'(t \trans[a] \mu)$ be a positive premise in $\rho'(\rtop(B', r , \phi'))$ 
 then $t \trans[a] \mu$ is a positive premise of a rule in $B'$ which does not belong
 to the image of $\phi'$. Then
 $\mu$ is such that $\rho(\mu) = \mu$ and this implies $\rho'(\mu) = \mu$.
 To prove that $\rho'(t)$ is a variable there are 2 cases to investigate:

\begin{itemize}
 \item $t \trans[a] \mu \in \text{top}(B,r,\phi)$. 
 Then $\sigma(t \trans[a] \mu) \in H$ and because $\frac{H}{c}$ is in  \nxmutt\ format, 
 then $\sigma(t)$ is a variable.
 Therefore $\sigma \rho'(t) = \sigma(t)$ and then $\rho'(t)$ is a variable.
 \item $t \trans[a] \mu \not \in \text{top}(B,r,\phi)$.
  Then there is a rule $b$ s.t. $\phi(b) = t \trans[a] \mu$.
  Since  $t \trans[a] \mu$ does not belong to the image of $\phi'$ we have that 
  $b \not  \in B'$. By $B'$ and the construction of $\rho$
  we have that $\rho^k(t)$ is a variable for all $k\geq 0$. 
  Then ($\rho'$\ref{unification:iii}) ensures that $\rho'(t)$ is a variable.  
\end{itemize}

This shows that the positive premises also fulfill the requirements of the \nxmuft\ format.

We proceed with the quantitative premises. 
Let $(\theta (Y) \gtgeq p) \in \qtop(B', r, \phi')$ with $\theta \in \openDTerms$.
By the same reasoning as applied for the target of the conclusion we get $\rho'(\theta) \in \openDTerms$. 
In addition, $\rho(y) = y$ for all $y \in Y$ because they do not appear in the left-hand side of a conclusion, and hence $\rho'(y) = y$.
Thus, $\rho'(\theta (Y) \gtgeq p)$ has the proper form.

Syntactical restriction for positive and quantitative premises  and conclusion are satisfied.
Besides, there is no restriction for negative premises, therefore $s$ is in \nxmuft\ format 
and then $s \in R'$.

 For all positive premises $c'\in \sigma(\rtop(B', r, \phi'))$ the rule $\frac{H}{c'}$
 is in $\nxmutt$ and it is provable in $R$ by a proof sub-structure smaller 
 than $(B, r, \phi)$. Thus, by induction we get that these rules are provable in $R'$.
 Applying Lemma~\ref{lemma:closure} on these rules and $s$ shows that $\frac{H}{c}$ is provable in $R'$.
\end{proof}

\begin{definition}
We say that a variable $x$ occurs \emph{free} in a rule $r$ if it
occurs in $r$ but not in the source of the conclusion nor in 
$W_j$ with $\theta_j(W_j) \gtrless_j q_j \in \qprem{r}$. 
\remarkDG{Does a instantiable dirac distribution $\delta_x$ as part of $\theta_j$ not also bind a variable and this should be included in the def?}
We say that a distribution variable $\mu$ occurs \emph{free} in a rule $r$ if it occurs in $r$ 
but not in the target of a positive premise. 
\end{definition}

\begin{definition}
 A PTSS $P = (\Sigma, A, R)$ is in \emph{pntree format} if
 all rules in $R$ are well-founded \nxmuft\ rules without free variables. 
\end{definition}

\begin{lemma}\label{lemma:pntree}
  Let $P = (\Sigma,A, R)$ be a PTSS in \nxmuft\ format.
  Then there is a PTSS $P' = (\Sigma, A, R')$ in
  pntree format such that for every closed transition rule $\dedrule{H}{c}$
  with only negative premises,
  $P \proves \dedrule{H}{c}$ iff $P' \proves \dedrule{H}{c}$
\end{lemma}
\begin{proof}
 Let $P' = (\Sigma, A, R')$ such that $R'$ is the set of provable rules 
 from $P$ in pntree format. 
 By Lemma~\ref{lemma:provability}, the right to left implication holds. 

 For the left to right implication we proceed by induction.
%
%
 Let $\dedrule{H}{c}$ be closed with $H$ containing negative literals
 only. Let $\dedrule{H}{c}$ be provable from $P$,
 i.e.\ $\dedrule{H}{c}\in R^\proves$.  Then either $c\in H$, $c$ is a
 valid closed quantitative literal, or there is a rule $r$ and a
 substitution $\rho$ such that $\rho(\conc{r}) = c$ and, for all
 premises $c' \in \rho(\pprem{r})$, $\dedrule{H}{c'}\in R^\proves$.
 Then $\dedrule{H}{c'}\in {R'}^\proves$ either trivially or by
 induction.

 Because $r$ is \nxmuft\ format, $r$ has the form
 \[
 \ddedrule%
      {
	{\textstyle
         \bigcup_{m\in M}
           \{ w_m^{\vec{z}} \trans[a_m]\mu_m^{\vec{z}} : \vec{z}\in \mathcal{Z} \} \ \cup \
         \bigcup_{n\in N}
           \{ t_n(\vec{z})\ntrans[b_n] : \vec{z}\in \mathcal{Z} \} \  \cup \
	  }
         {\textstyle
	    \{ \theta_l(Y_l)\gtgeq_{l,k} p_{l,k} : l\in L, k \in K_l \}
	  }
      }
      {f(x_1,\ldots,x_{\rank(f)}) \trans[a] \theta }
 \]
 where each $w_m^{\vec{z}}$ is a variable in $\vec{z}$.

 Let $G$ be the variable dependency graph associated to
 $\pprem{r}\cup\qprem{r}$.
 From $r$, we construct a rule $r'\in S$ as follows.
 Let $\mu_m^{\vec{z}}$ be the target of a positive premise such that
 there is no backward path in $G$ from a vertex $\mu_m^{\vec{z}}$ to
 some vertex $x_i$, with $i\in\{1,\ldots,x_{\rank(f)}\}$.  Notice
 that, by the symmetry requirements in Def.~\ref{def:ntmufxnu}, this
 happens for all $\mu_m^{\vec{z}}$ with $\vec{z}\in\mathcal{Z}$.
 We first obtain a rule $r''$ by
 \begin{inparaenum}[(i)]
 \item%
   replacing variables $w_m^{\vec{z}}$ and $\mu_m^{\vec{z}}$ by
   $\rho(w_m^{\vec{z}})$ and $\rho(\mu_m^{\vec{z}})$, respectively,
   and
 \item%
   replacing every free variable $\zeta$ in $\theta_l$ and $\theta$ by
   $\rho(\zeta)$.
 \end{inparaenum}
 The resulting rule $r''$ does not have free variables and it is a
 substitution instance of $r$, so $r''$ is provable from $P$.
 To obtain $r'$, replace each closed positive premise
 $\rho(w_m^{\vec{z}}\trans[a]\mu_m^{\vec{z}})$ by $H$.  Since,
 $w_m^{\vec{z}}\trans[a]\mu_m^{\vec{z}}$ is a positive premise of $r$,
 $\frac{H}{\rho(w_m^{\vec{z}}\trans[a]\mu_m^{\vec{z}})}\in R^\proves$.
 Then $r'$ is also provable from $P$.

 Notice that the resulting rule $r'$ is in \nxmuft\ format without
 free variables.  Morever, $r'$ is well-founded since any dependency
 backward chain ends in a vertex $x_i$.  Hence $r'$ is a pntree rule
 and therefore $r'\in R'$.

 Let $p\in\prem{r'}$. Then either $p\in H$ (and hence $p$ is closed)
 or $p\in\prem{r}$. In any case, $\frac{H}{\rho(p)}\in {R'}^\proves$
 (if $p\in\prem{r}$, it follows by induction). Therefore
 $\frac{H}{\rho(\conc{r'})}\in {R'}^\proves$.  Since
 $\rho(\conc{r'})=\rho(\conc{r})=c$, $\frac{H}{c}\in {R'}^\proves$.
\end{proof}

\begin{theorem}\label{th:ntumufxtPNTree}
 Let $P = (\Sigma, A, R)$ be a PTSS in \ntmufxt\ format. There is a 
 PTSS $P' = (\Sigma, A, R')$ in pntree format that is transition equivalent to $P$.
\end{theorem}

The proof of Theorem~\ref{th:ntumufxtPNTree} follows by applying
Lemmas~\ref{lemma:ntmuft}, \ref{lemma:nxmuft}, \ref{lemma:pntree},
and~\ref{lemma:sameNegRules}, in that order.

Let $P$ be a stratifiable PTSS in \ntmufxt\ format and let $S$ be its
stratification.
If $r$ is a provable rule from $P$, conditions (i) and (ii) in
Def.~\ref{def:stratification} also hold for stratification $S$ in rule
$r$.  (This can be shown by induction.)
Then, $S$ is also a stratification for the PTSS $P'$ in pntree format
obtained as in Theorem~\ref{th:ntumufxtPNTree}.
Since pntree rules are well-founded \ntmuft\ rules, from
Theorems~\ref{th:congruence} and~\ref{th:ntumufxtPNTree}, we have the
following corollary.

\begin{corollary}  
 If $P$ is a stratifiable PTSS in $\ntmufxt$\ format, $\bisim$ is a congruence for all operators in $P$.
\end{corollary}
\remarkDG{Even if it is a straight forward conclusion from the earlier propositions should we really just use corollary for this main theorem/main result of the paper (beside the technical machinery developed)? Maybe some of the earlier theorems should become only propositions?}

To conclude the section, we remark that negative premises cannot be
reduced to variables.
Following the nomenclature of~\cite{FokkinkvanGlabbeek96}, we say that
a rule is in \emph{simple pntree format} if it is in pntree format and
all its negative premises have the form $x \ntrans[a]$.
It turns out that the pntree format (and hence also the
\ntmufxt\ format) is strictly more expressive than simple pntree
format.
We will not dwell on this since example and rationale of the
difference of expressiveness in the non-probabilistic case applies
mutatis mutandi to our case (see~\cite{FokkinkvanGlabbeek96}).

 
\section{Concluding remarks}

We introduced the rule format \ntmufxt\ which enriches
\ntmufxnu~\cite{DL-fossacs12} by allowing distribution terms to appear
in quantitative premises and conclusions of rules.
We showed that it ensures that bisimulation equivalence is a
congruence for operators of well-founded PTSSs.  On proving this, we
corrected a mistake introduced in~\cite{DL-fossacs12}.
The richer syntactic structure of the quantitative premises and the
conclusion of the rules allows us to define a reduction of
\ntmufxt\ PTSSs to a transition equivalent PTSS consisting of only
pntree rules.
This construction confirms that the well-foundedness requirement in
\ntmufxt\ is not necessary to guarantee that bisimilarity is a
congruence.


We already know that the \ntmufxt\ format is equally expressive if
restricted to quantitative premises of the form $\theta(Y)> q$ with
$q\in [0,1]\cap\Q$.  
\remarkDG{Do you want to say that without loss of expressivity one can restrict $\gtgeq$ to only $>$ and express $\theta(Y)\ge q$ by $\{{\theta(Y) > q} \mid {q\in \Q \land q < r}\}$ ?}%
\remarkPRD{Not only that, but also that it suffices that $q$ is a rational number (instead of real). BTW, I corrected $\cup$ by $\cap$.}%
\remarkDG{Not required for the paper but just for my understanding: What is the reason that $\mathbb Q$ suffices? Is this following Dedikind-cut argumentation?}
\remarkPRD{Yes. For the case of $\theta(Y) > r$, we need to introduce an (denumerably) infinite number of rules. (This remark is also for me to remember how is the encoding.)}
However, we do not know whether distribution
terms are really needed.  We actually suspect that they are, and
hence, that the \ntmufxt\ format is strictly more expressive than the
\ntmufxnu\ format. 

Pntree rules are nearly ruloids \cite{BloomIM95:jacm} except that
negative premises may still contain non-variable terms. The decomposition 
method of \cite{Bloom:2004:PFD:963927.963929,Gebler:2012:phml_lt_sos} to develop
modular compositional proof systems can be adapted to pntree
rules by applying the negation-as-failure semantics for the logical characterization of
negative premises of pntree rules.
This will allow us to derive expressive congruence formats for probabilistic behavioral equivalences from their logical characterization in a structured way, following the approach of \cite{Bloom:2004:PFD:963927.963929}.

Both~\cite{DL-fossacs12} and this work have opened a new way of
thinking about probabilistic transition system specifications.  One of the
nicest things is that the \ntmufxt\ follows quite closely the
structure of non-probabilistic formats (particularly, \ntyfxt).
Hence, many ideas for further work can be borrowed from the
non-probabilistic setting.

\bibliographystyle{eptcs}
\bibliography{pntree}

\end{document}

\appendix

\section{Proof of the congruence theorem} \label{ap:congruence}

Before proving the theorem we need to show some auxiliary lemmas. 
%

\begin{proof}[Proof of Lemma~\ref{lemma:ntmuft}]
   Let $\Sigma = (F,\rank)$. 
  For all rule 
  \[r = \ddedrule{\prem{r}}
        {x \trans[a] \theta }\in R\]
  in \ntmuxt\ format,
  $\theta$ is such that $\theta \in \openDTerms$ then either
  $\theta \in \PVar \cup \DVar$ or 
  $\theta = {\textstyle \sum_{i\in I} p_i (\prod_{n_i \in N_i} \theta_{n_i}) \circ g_i^{-1}} $ 
  where $p_i \in (0,1]$ with $\sum_{i\in I} p_i = 1$,
  each $g_i : \closedTerms^{N_i} \rightarrow \closedTerms$ is a  function,
  and $\theta_{n_i}\in \openDTerms$.
  We only take in account the second case, the first one is similar to this one. 
  W.l.o.g. we suppose that each product $(\prod_{n_i\in N_i} \theta_{n_i})$ has at most 
  one term $\theta_{n_i}$ such that $\theta_{n_i} = \delta_t$ with $x \in \Var(t)$.  
  Then for each $r$ and function symbol $f \in F$ we define the  open substitution $\rho_f$ by:
  \begin{center}
    \begin{tabular}{l!{\ \ }l}
      $\rho_f(x) = f(z_1, \cdots, z_{\rank(f)})$ &
      if $x$ is in the source of the conclusion of $r$ and \\
      & $z_1, \cdots, z_{\rank(f)}$  are variables that do not occur in $r$.  \\
      $\rho_f(x) = x$ & otherwise.
    \end{tabular}
  \end{center}
%

Let $R'$ be a set of rules such that if $r
\in R$ and $r$ is in \ntmuft\ format, then $r \in R'$.
In addition, for all rule $r$ in \ntmuxt\ format and every symbol function $f$, 
 we add the rule $r'$ in $R'$ such that $r'$ is defined by
$r' = \rho_f(r)$. 

  The set of rules $R'$ is in \ntmuft\ format. 
  It is not hard to prove that if $S$ is a stratification for $P$,
  then it is also a stratification for $P'$.  By transfinite induction
  suppose ${\trans_{P_{\beta'}}} = {\trans_{P'_{\beta'}}}$ holds for
  all $\beta' < \beta$ and ${\trans_{P_{\beta,j'}}} =
  {\trans_{P'_{\beta,j'}}}$ for all $j'<j$, then we have to prove that
  ${\trans_{P_{\beta,j}}} = {\trans_{P'_{\beta,j}}}$ to ensure
  ${\trans_{P}} = {\trans_{P'}}$.
%
  We only show ${\trans_{P_{\beta,j}}} \subseteq
  {\trans_{P'_{\beta,j}}}$, the case ${\trans_{P_{\beta,j}}} \supseteq
  {\trans_{P'_{\beta,j}}}$ is analogous.

  Suppose $\psi = f(t_1, \dots, t_n) \trans[a] \pi \in
  {\trans_{P_{\beta,j}}}$.
  Then, there is rule $r \in R$ and substitution $\rho$ such that
  $\conc{\rho(r)}=\psi$ and, for all $\phi\in\prem{\rho(r)}$,
  $\bigcup_{\beta'<\beta}{\trans_{P_{\beta'}}} \cup
  \bigcup_{j'<j}{\trans_{P_{\beta,j'}}} \models \phi$.
  If $\phi$ is positive then, by induction,
  $\bigcup_{\beta'<\beta}{\trans_{P'_{\beta'}}} \cup
  \bigcup_{j'<j}{\trans_{P'_{\beta,j'}}} \models \phi$.
  If $\phi \equiv t \ntrans[b]$ then, for all
  $\pi'\in\closedDTerms$,
  $t\trans[b]\pi' \notin \bigcup_{\beta'<\beta}{\trans_{P_{\beta'}}}$
  and therefore
  $t\trans[b]\pi' \notin
  \bigcup_{\beta'<\beta}{\trans_{P'_{\beta'}}}$.
  Besides, if $\phi\in\qprem{\rho(r)}$,
  $\bigcup_{\beta'<\beta}{\trans_{P'_{\beta'}}} \cup
  \bigcup_{j'<j}{\trans_{P'_{\beta,j'}}} \models \phi$
  trivially, as the model has not effect on its validity.
  Therefore, $\bigcup_{\beta'<\beta}{\trans_{P'_{\beta'}}} \cup
  \bigcup_{j'<j}{\trans_{P'_{\beta,j'}}} \models \prem{\rho(r)}$.

  If $r$ is in \ntmuft\ format, then $r \in R'$ and using $\rho(r)$ one
  can derive $\psi\in{\trans_{P'_{\beta,j}}}$
  If, instead, $r$ is in \ntmuxt\ format with
  $\conc{r} = x \trans[a] {\textstyle \sum_{i\in I} p_i ( \prod_{n_i \in N_i } \theta_{n_i}) \circ g_{i}^{-1} }$, 
  then take the rule
  \[r' = \ddedrule{\rho_f(\prem{r})}
           {\rho_f(x) \trans[a] 
           \rho_f({\textstyle \sum_{i\in I} p_i ( \prod_{n_i \in N_i } \theta_{n_i}) \circ g_{i}^{-1} }) } \in R'\]
  with $\rho_f(x) = f(z_1, \cdots, z_{\rank(f)})$ and define substitution $\rho'$ as $\rho'(z_k) = t_k$, for $1\leq
  k\leq\rank(f)$ and $\rho'(x) = \rho(x)$ for any other (term or
  distribution) variable.
  Notice that
  $\rho'(\prem{r'})=\rho'(\rho_f(\prem{r})) = \rho(\prem{r})$, so
  $\bigcup_{\beta'<\beta}{\trans_{P'_{\beta'}}} \cup
  \bigcup_{j'<j}{\trans_{P'_{\beta,j'}}} \models \rho'(\prem{r'})$.
  Moreover,
  $\rho'(f(z_1,\dots,z_{\rank(f)})) = f(t_1,\dots,t_{\rank(f)}) =
  \rho(x)$,
  then we only have to prove that the targets of $\psi$ and
  $\rho'(\conc{r'})$ are the same probability distribution.
  This is straightforward, applying the same reasoning used with 
  the premises of the rule.
\end{proof}

The rule $r$ is called \emph{pure} if it
is well-founded and  does not contain free variables. 
A PTSS $P$ is called \emph{pure} if all of its rules are pure.
The next lemma states that
well-founded and stratifiable PTSSs in
\ntmufxt\ format can be reduced to another that is also pure and
transition equivalent to the given one.

\begin{lemma}\label{lemma:semi-pure}
  Let $P=(\Sigma, A, R)$ be a stratifiable and well-founded PTSS in \ntmufxt\ format.
  Then there is a stratifiable PTSS $P'=(\Sigma, A, R')$ in pure
  \ntmufxt\ format that is transition equivalent with $P$.
  Moreover, if $P$ is in \ntmuft\ format then $P'$ is in
  \ntmuft\ format.
\end{lemma}
\begin{proof}
  This proof proceeds very much like the proof of lemma~\ref{lemma:ntmuft}.
  For every rule $r \in R$, if $r$ has no free variables then $r \in R'$.
  If $r \in R$ has free variables, then for each $r$ such that
  \[r = \ddedrule{\prem{r}}
         {t \trans[a] {\textstyle \sum_{i\in I} p_i (\prod_{n_i\in N_i}
         \theta_{n_i}) \circ g_i^{-1}}} \in R\]
  and all substitution $\hat{\rho}$ such that
  \begin{enumerate}
  \item%
    $\hat{\rho}(x) = t_x \in \closedTerms$, for all free term variable $x$ in $r$,
  \item%
    $\hat{\rho}(\mu) = \pi_\mu \in \closedDTerms$,
    for all free distribution variable  $\mu$ in $r$, and
  \item%
    $\hat{\rho}(x) = x$ otherwise,
  \end{enumerate}
  we define $r' \in R'$ by 
  \[r' = \ddedrule{\hat{\rho}(\prem{r})} 
           {t \trans \hat{\rho}  \left ({\textstyle \sum_{i\in I} p_i (\prod_{n_i\in N_i}
           \theta_{n_i}) \circ g_i^{-1}} \right) }\]
  One more time with do not take in account the case where the 
  target of the conclusion belongs to $\PVar \cup \DVar$.
  Note that the rules in $R'$ has not free variables.
  By transfinite induction
  suppose ${\trans_{P_{\beta'}}} = {\trans_{P'_{\beta'}}}$ holds for
  all $\beta' < \beta$ and ${\trans_{P_{\beta,j'}}} =
  {\trans_{P'_{\beta,j'}}}$ for all $j'<j$, then we have to prove that
  ${\trans_{P_{\beta,j}}} = {\trans_{P'_{\beta,j}}}$ to ensure
  ${\trans_{P}} = {\trans_{P'}}$.
  We only show
  ${\trans_{P_{\beta,j}}} \subseteq {\trans_{P'_{\beta,j}}}$, the case
  ${\trans_{P_{\beta,j}}} \supseteq {\trans_{P'_{\beta,j}}}$ is
  analogous.

  Suppose $t \trans[a] \pi \in {\trans_{P_{\beta,j}}}$.
  Then, there is a rule $r \in R$ with
  $\conc{r} = \hat{t} \trans[a]
       \sum_{i\in I} p_i (\prod_{n_i\in N_i} \theta_{n_i}) \circ g_i^{-1}$,
  and substitution $\rho$ such that
  $\rho(\hat{t} \trans[a]
        \sum_{i\in I} p_i (\prod_{n_i\in N_i} \nu_{n_i}) \circ g_i^{-1})
   = {t \trans[a] \pi}$
  and, for all $\phi\in\prem{\rho(r)}$,
  $\bigcup_{\beta'<\beta}{\trans_{P_{\beta'}}} \cup
  \bigcup_{j'<j}{\trans_{P_{\beta,j'}}} \models \phi$.
  Like before, we obtain that 
  $\bigcup_{\beta'<\beta}{\trans_{P'_{\beta'}}} \cup
  \bigcup_{j'<j}{\trans_{P'_{\beta,j'}}} \models \phi$ and, in particular,
  if it is negative
  $\bigcup_{\beta'<\beta}{\trans_{P'_{\beta'}}}\models\phi$.

  Define $\hat{\rho}$ by $\hat{\rho}(x) = \rho(x)$, for all free term
  variable $x$ in $r$, $\hat{\rho}(\mu) = \pi_\mu = \rho(\mu)$ with $\pi_\mu \in \closedDTerms$,
  for all free distribution variable $\mu$ in $r$, and $\hat{\rho}(x) = x$
  otherwise.
  Then take rule
  \[r' = \ddedrule{\hat{\rho}(\prem{r})} 
           {t \trans \hat{\rho}  \left ({\textstyle \sum_{i\in I} p_i (\prod_{n_i\in N_i}
           \theta_{n_i}) \circ g_i^{-1}} \right) }\]
  Notice that
  $\rho(\prem{r'})=\rho(\hat{\rho}(\prem{r})) = \rho(\prem{r})$, so
  $\bigcup_{\beta'<\beta}{\trans_{P'_{\beta'}}} \cup
  \bigcup_{j'<j}{\trans_{P'_{\beta,j'}}} \models \rho(\prem{r'})$.
  Moreover, the source of both $\rho(\conc{r})$ and
  $\hat{\rho}(\conc{r'})$ is $\rho(\hat{t})$.
  Then we only have to prove that their targets are the same
  probability distribution.  
  This is straightforward applying the same reasoning 
  used for the premises. 
%
\end{proof}

The proof the following lemma follows from straightforward calculations.

\begin{lemma}\label{lemma:eq:dist:equiv}%
  Let ${\relR} \subseteq \closedTerms \times \closedTerms$ be an
  equivalence relation.  For all (discrete) distributions
  $\pi,\pi'\in\Delta(\closedTerms)$, $\pi \relR \pi'$ iff for every
  equivalence class $Q\in \closedTerms/{\relR}$, $\pi(Q) = \pi'(Q)$.
\end{lemma}

We say that a stratification $S : \PTr \to \alpha$ is
\emph{target-independent} if the stratum does not depend of the
target, that is, if for all $t\in\closedTerms$, $a\in A$, and
$\pi,\pi'\in\Delta(\closedTerms)$,
$S(t\trans[a]\pi)=S(t\trans[a]\pi')$.  In this case, we unambiguously
write $S(t,a)$ instead of $S(t\trans[a]\pi)$.

The proof of the next lemma follows mutatis mutandi the proof of
Lemma~4.8 in~\cite{Groote93}.

\begin{lemma}\label{lm:ntmufnu:target:indep}%
  Let $P$ be a stratifiable PTSS in \ntmufxt\ format. Then $P$
  has a target independent stratification.
\end{lemma}

Relation $\relR_P$, introduced in Def.~\ref{def:RP}, is central to the
proof of Theorem~\ref{th:congruence}. By construction, $\relR_P$ is
preserved by all context (see Lemma~\ref{lemma:sust_de_var} whose
proof follows by straightforward induction on the structure of $t$.)

\begin{definition}\label{def:RP}%
  Let $\Sigma = (F, \rank)$ be a signature and $P=(\Sigma, A, R)$ a
  PTSS with an associated transition relation. The relation ${\relR_P}
  \subseteq \closedTerms \times \closedTerms$ is the smallest 
  relation satisfying:
  \begin{enumerate}[(i)]
  \item%
    ${\bisim} \subseteq {\relR_P}$, and
  \item%
    for all $f \in F$, %
    $\left(\bigwedge_{k=1}^{\rank(f)} \ {u_k \relR_P v_k}\right)
    \ \limp \ 
    f(u_1,\cdots,u_{\rank(f)}) \relR_P f(v_1,\cdots,v_{\rank(f)})$.
  \end{enumerate}
\end{definition}

 The proof of the next lemma follows by straightforward induction on the
 structure of $t$.

\begin{lemma}\label{lemma:sust_de_var}
 Let $t \in \openTerms$ and let $\rho, \rho' : \TVar \to \closedTerms$ be two 
 substitutions such that for all $x \in V$ holds $\rho(x) \relR_P \rho'(x)$.
 Then  $\rho(t) \relR_P \rho'(t)$.
\end{lemma}

Besides, a similar theorem holds for distribution variables:
\begin{lemma}\label{lemma:sust_de_Dvar}
 Let $\theta \in \openDTerms$ and let $\rho, \rho' : \TVar \cup \PVar \to \closedTerms \cup \closedDTerms$ 
 be two substitutions such that for all variable $\zeta \in \Var(\theta)$ holds $\rho(\zeta) \relR \rho'(\zeta)$.
 Then  $\rho(\theta) \relR \rho'(\theta)$.
\end{lemma}

\begin{proof}
  The proof is by structural induction over $\theta \in \openDTerms$.
  The base case is straightforward. 
  Suppose then that for all $\zeta \in \Var({\textstyle \sum_{i\in I} p_i (\prod_{n_i \in N_i} \theta_{n_i}) \circ g_i^{-1}})$ 
  holds $\rho(\zeta) \relR \rho'(\zeta)$,  we should prove that

  \begin{equation}\label{eq:equivalence:of:target}
    \rho\left(
      {\textstyle \sum_{i\in I} p_i (\prod_{n_i \in N_i} \theta_{n_i}) \circ g_i^{-1}}
      \right)
    \relR_P
    \rho'\left(
      {\textstyle \sum_{i\in I} p_i (\prod_{n_i\in N_i} \theta_{n_i}) \circ g_i^{-1}}
      \right)
  \end{equation}

  By induction hypothesis we have that $\rho(\theta_{n_i}) = \rho'(\theta_{n_i})$.
  By Lemma~\ref{lemma:sust_de_var}, for all $g_i$ and
  $u_h,u'_h\in\closedTerms$ such that $u_h \relR_P u'_h$, $1\leq h\leq
  N_i$,
  $g_i(u_1,\cdots,u_{N_i}) \relR_P g_i(u'_1,\cdots,u'_{N_i})$.
  As a consequence, for every $d\in \closedTerms/{\relR}$,
  $g_i^{-1}(d) = \biguplus_{k\in K} d^k_1\times\cdots\times d^k_{N_i}$,
  where $d^k_h\in \closedTerms/{\relR}$, for all $k\in K$ and $1\leq
  h\leq N_i$.
  Then, for each $i\in I$, and equivalence class $d\in \closedTerms/{\relR}$ we calculate
  \begin{align}
    \textstyle
    ((\prod_{n_i\in N_i} \rho(\theta_{n_i})) \circ g_i^{-1}) (d)
    &=\textstyle
    (\prod_{n_i\in N_i} \rho(\theta_{n_i})) \left(\biguplus_{k\in K} d^k_1\times\cdots\times d^k_{N_i}\right)
    \notag
    \\
    &=\textstyle
    \sum_{k\in K} (\prod_{n_i\in N_i} \rho(\theta_{n_i})) (d^k_1\times\cdots\times d^k_{N_i})
    \notag
    \\
    &=\textstyle
    \sum_{k\in K} \rho(\theta_1)(d^k_1) \cdot\ldots\cdot \rho(\theta_{N_i})(d^k_{N_i})
    \notag
    \\
    &=\textstyle
    \sum_{k\in K} \rho'(\theta_1)(d^k_1) \cdot\ldots\cdot \rho'(\theta_{N_i})(d^k_{N_i})
    \notag
    \\
    &=\textstyle
    ((\prod_{n_i\in N_i} \rho'(\theta_{n_i})) \circ g_i^{-1}) (d)
    \tag{calculations are as before}
  \end{align}
  From here, and using Lemma~\ref{lemma:eq:dist:equiv}, it is
  straightforward to calculate (\ref{eq:equivalence:of:target}),
  thus concluding the proof.
\end{proof}

The proof of Theorem~\ref{th:congruence} proceeds by showing that the relation
$\relR_P$ of Def.~\ref{def:RP} is a bisimulation.  As a consequence
${\bisim} = {\relR_P}$ since ${\bisim}\subseteq {\relR_P}$.
Lemma~\ref{lemma:sust_de_var} finally ensures congruence.

\begin{proof}[Proof of Theorem~\ref{th:congruence}]
 
  By Lemmas \ref{lemma:ntmuft} and \ref{lemma:semi-pure} we can assume $P$
  is in pure \ntmufnu\ format.  In addition, $P$ is stratifiable. Then
  there is a transition relation $\trans_P$ associated with $P$.
  We have to show that for all 
  $f \in F, u_1, \cdots u_{\rank(f)}, v_1, \cdots, v_{\rank(f)} \in \closedTerms$
  it holds that:
  \[%
    1 \leq k \leq \rank(f) : u_k \bisim v_k \implies 
    f(u_1, \cdots, u_{\rank(f)}) \bisim f(v_1, \cdots, v_{\rank(f)})
  \]
  In order to do so, we show that $\relR_P$ (see Def.~\ref{def:RP}) is
  a bisimulation relation.  If this is the case, then ${\relR_P}
  \subseteq {\bisim}$.  Since ${\bisim} \subseteq {\relR_P}$, by
  definition of ${\relR_P}$, we have that ${\relR_P} = {\bisim}$ and
  hence ${\bisim}$ is a congruence as consequence of $\relR_P$ being a
  congruence.

  First, notice that $\relR_P$ is an equivalence relation (this can be
  proved by straightforward induction on the construction of
  $\relR_P$).  Then, to show that $\relR_P$ is a bisimulation we have
  to check that $\relR_P$ satisfy the transfer property, i.e., that
  for all $t,t'\in \closedTerms$, $a\in A$, and
  $\pi\in\closedDTerms$,
  \[%
    t \relR_P t' \text{ and } t \trans[a] \pi \ \ \implies \ \
    \exists \pi' : t' \trans[a] \pi' \text{ and } \pi \relR_P \pi'.
  \]

%

  If $t \relR_P t'$ is a consequence of Def.~\ref{def:RP}(i) holds, then $t \bisim
  t'$ and hence $t \trans[a] \pi$ implies that there exists $\pi' \in
  \closedDTerms$ such that $t' \trans[a] \pi'$ and $\pi \bisim
  \pi'$.
  Since ${\bisim} \subseteq {\relR_P}$, $\closed{\relR_P}(c)$ implies
  $\closed{{\bisim}}(c)$, for all $c \subseteq \closedTerms$.
  Then it is also the case that ${\pi} {\relR_P} {\pi'}$ which proves the
  transfer property for case (i) in Def.~\ref{def:RP}.

%

  Now, suppose that $t \relR_P t'$ is a consequence of Def.~\ref{def:RP}(ii),
  that is $t= f(v_1, \cdots, v_{\rank(f)})$, $t' = f(w_1,
  \cdots, w_{\rank(f)}) $, with $f \in F$ and $v_k \relR_P w_k$ for $
  1 \leq k \leq \rank(f)$.

  By Lemma~\ref{lm:ntmufnu:target:indep}, let $S:\PTr\to\alpha$
  be a target-independent stratification of $P$.  Notice that it sufficies to show that, for all $\gamma$:
  \begin{quote}
    If
    $S(f(v_1,\cdots,v_{\rank(f)}),a)+S(f(w_1,\cdots,w_{\rank(f)}),a)=\gamma$,
    then,
    \hfill(IH) 
    \par %
    \quad\ %
    $f(v_1,\cdots,v_{\rank(f)}) \trans[a] \pi \in {\trans_P}$ and $v_k
    \relR_P w_k$ for $ 1 \leq k \leq \rank(f)$ implies that \par %
    \qquad\ \ %
    there exists $\pi' \in \closedDTerms$ s.t.\
    $f(w_1,\cdots,w_{\rank(f)}) \trans[a] \pi' \in {\trans_P}$ and
    $\pi \relR_P \pi'$.
  \end{quote}

  We proceed by induction.
  Suppose  (IH) holds for all
  $\gamma'<\gamma$, then the validity of (IH) for $\gamma$ follows
  immediately if for all $1\leq\beta<\alpha$ and $j\geq 0$,
  \begin{quote}
    If
    $S(f(v_1,\cdots,v_{\rank(f)}),a)+S(f(w_1,\cdots,w_{\rank(f)}),a)=\gamma$,
    then,
    \hfill(IH') 
    \par %
    \quad\ %
    $f(v_1,\cdots,v_{\rank(f)}) \trans[a] \pi \in {\trans_{P_{\beta,j}}}$
    and $v_k \relR_P w_k$ for $ 1 \leq k \leq \rank(f)$ implies that
    \par %
    \qquad\ \ %
    there exists $\pi' \in \closedDTerms$ s.t.\
    $f(w_1,\cdots,w_{\rank(f)}) \trans[a] \pi' \in {\trans_P}$ and
    $\pi \relR_P \pi'$.
  \end{quote}

  To prove (IH') in general, we proceed by induction on the stratum
  $\beta$ and, within it, by induction on $j$.
  So, suppose (IH') holds for all $\beta'<\beta$ and for $j'<j$
  whenever $\beta'=\beta$.
  Assume
  $S(f(v_1,\cdots,v_{\rank(f)}),a)+S(f(w_1,\cdots,w_{\rank(f)}),a)=\gamma$,
  and
  $f(v_1,\cdots,v_{\rank(f)})\trans[a]\pi \in {\trans_{P_{\beta,j}}}$.
  Then, by Def.~\ref{def:assoc_with}, there must exist a rule (call it $r$)
  \[\ddedrule%
      {\textstyle
         \bigcup_{m\in M}
           \{ t_m(\vec{z})\trans[a_m]\mu_m^{\vec{z}} : \vec{z}\in \mathcal{Z} \} \cup
         \bigcup_{n\in N}
           \{ t_n(\vec{z})\ntrans[b_n] : \vec{z}\in \mathcal{Z} \}
         \{ \theta_l(Y_l)\gtgeq_{l,k} p_{l,,k} : l\in L, k \in K_l\}
      }
      {f(x_1,\ldots,x_{\rank(f)}) \trans[a] \theta}
  \]
  and a proper substitution $\rho$ such that, for all $1\leq k\leq\rank(f)$,
  $\vec{z}\in \mathcal{Z}$, $m\in M$, $n\in N$, $l\in L$  and $k \in K_l$
  \begin{enumerate}[($\rho$1)~]
  \item\label{rho:i}%
    $\rho(x_k)=v_k$,
  \item\label{rho:ii}%
    $\pi = \rho(\theta)$,
  \item\label{rho:iii}%
    $\bigcup_{\beta'<\beta}{\trans_{P_\beta'}} \cup \bigcup_{j'<j}{\trans_{P_{\beta,j'}}} \models \rho(t_m(\vec{z})\trans[a_m]\mu_m^{\vec{z}})$,
  \item\label{rho:iv}%
    $\bigcup_{\beta'<\beta}{\trans_{P_\beta'}} \models
     \rho(t_n(\vec{z})\ntrans[b_n])$, and
  \item\label{rho:v}%
   $\rho(\theta_l(Y_l))\gtgeq_{l,k} p_{l,k}$.
  \end{enumerate}

  We need to prove the existence of a proper substitution $\rho'$ that, using rule $r$ and
  induction, allows us to conclude the transfer property in (IH').
  Notice that requiring that the substitution is proper creates a dependency between 
  the variables in $\Var(\theta_l)$ and $Y_l$, i.e., we cannot define 
  the variables $Y_l$ before the variables in $\Var(\theta_l)$. 
  To satisfy this dependency the definition of substitution $\rho'$ 
  is done using the variable dependency graph $G$ 
  of the positive and quantitative premises in $r$. 

  In addition, 
  to define $\rho'$ we need to partition the set $\mathcal{Z}$ in such a way that it
  allows us to cover all required pairs on relation $\relR_P$ between
  instances through substitution $\rho$ and instances through substitution $\rho'$.
  For this, define
  $\rho(Y_l)/{\relR_P} = \{\rho(Y_l)\cap[t]_{\relR_P} \mid t\in\closedTerms\}$.
  That is, we obtain the partition of $\rho(Y_l)$ induced by the
  equivalence relation $\relR_P$.
  For each $d\in\rho(Y_l)/{\relR_P}$ define $d^\uparrow$ as the unique
  $d^\uparrow \in \closedTerms/{\relR_P}$ s.t.\ $d \subseteq
  d^\uparrow$.
Now we define $\partZ$, a partition of $\mathcal{Z}$,  such that
for all $l\in L$ and $d_l \in \rho(Y_l)/{\relR_P}$ there is a
$Z_{\prod_{l\in L}d_l} \in \partZ$ such that 
\begin{enumerate}[($\partZ$1)~]
 \item \label{part_cond:i} %
 $\exists \vec{z} \in Z_{\prod_{l\in L}d_l} : \rho(\vec{z}) \in {\prod_{l\in L} d_l \times \prod_{k=1}^{\rank(f)}\{v_k\}}$
 \item \label{part_cond:ii}%
 $\cardof{Z_{\prod_{l\in L}d_l}}  \geq \max_{l\in L}\cardof{d_l^\uparrow}$
\end{enumerate}
  Notice that the cardinality of the $Y_l$'s and $L$
  (restriction~\ref{item:conditions_on_cardinality} in
  Def.~\ref{def:ntmufxnu}) guarantees that such partition can be
  constructed.

  Now, we will define a proper closed substitution $\rho'$ such that for any 
  variable $\zeta \in \Var(r)$ the following holds: 

\begin{enumerate}[($\rho'$1)~]
  \item\label{rhoprime:i}%
    if $\zeta$ is a term variable s.t. $\zeta = x_k$ for some $1\leq k\leq \rank(f)$ 
    then $\rho(\zeta) \relR_P \rho'(\zeta)$;
%
%
   \item\label{rhoprime:targetCond} %
    if $\zeta = \mu_m^{\vec{z}}$ with $m \in M$ and $\vec{z} \in \mathcal{Z}$,
    then $\rho'(t_m(\vec{z})\trans[a_m]\mu_m^{\vec{z}}) \in {\trans_P}$.  
    Moreover if $\mu^{\vec{z}}_m \in \Var(\theta)$
    or $\mu^{\vec{z}}_m$ appears in a quantitative premise,
    then $\rho'(\mu^{\vec{z}}_m) \relR_P \rho(\mu^{\vec{z}}_m)$.
  \item\label{rhoprime:iii} %
  if $\zeta = y \in Y_l$, then $Y_l$ is such that
  \footnote{Notice that $\forall y' \in Y_l : \nvdg(y) = \nvdg(y')$.
  This allows us to define the properties over all the set $Y_l$}: 
\begin{enumerate}
 \item \label{rhoprime:iii:d}%
      $\rho'(\pi_l(Z_{\prod_{l\in L}d_l})) =
      d_l^\uparrow \cap \support(\rho'(\theta_l))  $ (here, $\pi_l$ indicates the $l$th projection),
 \item \label{rhoprime:iii:a} 
          if $y \in \Var(\theta)$ then $\rho(y) \relR_P \rho'(y)$
 \item \label{rhoprime:termCond}%
    for all $t_m(\vec{z})$ with $\vec{z} = \tuple{z_0, \ldots, z_{|\vec{z}|}}$, 
     there is $\vec{z'} = \tuple{z'_0, \ldots, z'_{|\vec{z}|}}$ such that 
    and for all $i \in \{0, \ldots, |\vec{z}|\}$ if $\nvdg(z_i) \leq \nvdg(y)$ 
    then $\rho'(z_i) \relR_P \rho(z'_i)$.
    Moreover if $\mu^{\vec{z}}_m \in \Var(\theta)$
    or $\mu^{\vec{z}}_m$ appears in a quantitative premise,
    then $\rho'(z_i) \relR_P \rho(z_i)$.
\end{enumerate}
 \end{enumerate}

 We define $\rho'(x)$ by induction over $n = \nvdg(x)$,
 from which the properties above can be proven straightforwardly. 

 If $n=0$ then $\nvdg(\zeta) =0$, in view that the format is pure, 
 we have  $\zeta = x_i$ for $i \in \{1, \ldots, \rank(f)\}$. Then define $\rho'(\zeta) = w_i$.
 Notice that condition $(\rho'\ref{rhoprime:i})$ is 
 satisfied and the other conditions are not imposed over this kind of variables.  

 Suppose $\rho'$ is defined for $\zeta'$ such that
 $0 \leq \nvdg (\zeta') < n$, we define $\rho'$ for $\zeta$ such that $\nvdg(\zeta) = n$. 
  We have two cases to analyze. 
  In the first case we take $\zeta = \mu^{\vec{z}}_m$.
  Because $\nvdg(\mu^{\vec{z}}_m) = n$,
  for all $x' \in \Var(t_m(\vec{z}))$, $\nvdg(x')<n$.  
  By condition $(\rho'\ref{rhoprime:termCond})$, 
  we can ensure that there is $\vec{z'}$ such that 
  $\rho'(t_m(\vec{z})) \relR_P \rho(t_m(\vec{z'}))$.
   By ($\rho$\ref{rho:iii}),
  $\rho(t_m(\vec{z'}))\trans[a_m]\rho(\mu_m^{\vec{z'}}) \in
  \bigcup_{\beta'<\beta}{\trans_{P_\beta}} \cup \bigcup_{j'<j}{\trans_{P_{\beta,j}}}$,
  since $t_m(\vec{z'})\trans[a_m]\mu_m^{\vec{z'}}$ is also a
  positive premise of $r$.
  Now, by the definition of $\relR_P$, two cases arise:
  \begin{enumerate}
  \item%
    $\rho(t_m(\vec{z'})) \bisim \rho'(t_m(\vec{z}))$.
    Then there exists $\pi\in\closedDTerms$
    s.t.\ $\rho'(t_m(\vec{z}))\trans[a_m]\pi$ and
    $\rho(\mu_m^{\vec{z'}}) \bisim \pi$ (and hence also
    $\rho(\mu_m^{\vec{z'}}) \relR_P \pi$).
  \item%
    There is a function name $g\in F$ and terms
    $u_k,u'_k\in\closedTerms$, $1\leq k\leq\rank(g)$, such that
    \begin{itemize}
    \item%
      $\rho(t_m(\vec{z'})) = g(u_1,\cdots,u_{\rank(g)})$,
    \item%
      $\rho'(t_m(\vec{z})) = g(u'_1,\cdots,u'_{\rank(g)})$, and
    \item%
      $u_k \relR_P u'_k$ for all $1\leq k\leq\rank(g)$.
    \end{itemize}
    Furthermore, we know that
    $$S(\rho(t_m(\vec{z'})),a_m)+S(\rho'(t_m(\vec{z})),a_m) \leq
     S(f(v_1,\cdots,v_{\rank(f)}),a)+S(f(w_1,\cdots,w_{\rank(f)}),a).$$
    Now we can apply the first or second induction hypothesis and
    conclude that there exists some $\pi\in\closedDTerms$
    s.t.\ $\rho'(t_m(\vec{z}))\trans[a_m]\pi \in {\trans_P}$ and
    $\rho(\mu_m^{\vec{z'}}) \relR_P \pi$.
  \end{enumerate}
  Either way, define $\rho'(\mu_m^{\vec{z}}) = \pi$ and hence 
   \begin{itemize}
  \item
     $\rho'(t_m(\vec{z})\trans[a_m]\mu_m^{\vec{z}}) \in {\trans_P}$, and
   \item
     $\rho(\mu_m^{\vec{z'}}) \relR_P \rho'(\mu_m^{\vec{z}})$;
     in particular $\rho(\mu_m^{\vec{z}}) \relR_P \rho'(\mu_m^{\vec{z}})$
     whenever $\mu_m^{\vec{z}} \in \Var(\theta)$
    or $\mu^{\vec{z}}_m$ appears in a quantitative premise
    by condition ($\rho'$\ref{rhoprime:termCond}).
   \end{itemize}

 This definition ensures condition $(\rho'\ref{rhoprime:targetCond})$.

  In the second case we have that $\zeta \in Y_{l'}$ for some $l' \in L$. 
  In this case we define the substitution for all the set $Y_{l'}$ at the same time.
  For all term variable in $Y_{l'}$, define $\rho'$ such that
  for each $Z_{\prod_{l\in L}d_l}\in\partZ$
  \begin{itemize}
    \item%
      $\rho'(\pi_{l'} (Z_{\prod_{l\in L}d_l})) =
      d_{l'}^\uparrow \cap \support(\rho'(\theta_{l'}))$,
      and
    \item\label{cond:ii:b}%
      for all $\vec{z}\in \mathcal{Z}$, $m\in M$, if 
      $\mu^{\vec{z}}_m \in \Var(\theta)$
      or $\mu^{\vec{z}}_m \in \Var(\theta_{l'})$ for some $l' \in L$,
      then
      $\rho'(\vec{z}(l')) \in d_{l}^\uparrow \liff \rho(\vec{z}(l')) \in
      d_{l}$ for all $d_l\in\rho(Y_{l'})/{\relR_P}$,
    \item\label{cond:ii:c}%
      finally, if $y \in \Var(\theta)$
      then $\rho'(y) \in d^\uparrow \liff \rho(y) \in d$ for all $d \in\rho(Y_{l'})/{\relR_P}$.
  \end{itemize}
  Notice that this definition is possible because
  $\rho'$ is defined for all variable in $\Var(\theta_{l'})$ i.e. $\rho'(\theta_{l'}) \in \closedDTerms$, 
  $\cardof{Z_{\prod_{l\in L}d_l}} \geq \max_{l\in L}\cardof{d_l^\uparrow}$,
  $\mathcal{Z} = \diag{Y_l}_{l\in L} \times\prod_{k=1}^{\rank(f)}\{x_k\}$,
  conditions~\ref{item:conditions_on_cardinality}  
  and~\ref{item:conditions_on_nonrepeating_z} in
  Def.~\ref{def:ntmufxnu}, and the fact that variables in 
  $\theta$ are finite.
  This definition ensures properties $(\rho'\ref{rhoprime:iii:d})$ and  $(\rho'\ref{rhoprime:iii:a})$ 
  The property $(\rho'\ref{rhoprime:termCond})$ is ensured by 
  condition $(\partZ \ref{part_cond:i})$ and the induction hypothesis.

  The definition of $\rho'$ ensures that 
  $\rho'(t_m(\vec{z})\trans[a_m]\mu_m^{\vec{z}}) \in {\trans_P}$. 
  Now, we show that each $t_n(\vec{z})\ntrans[b_n]$ holds in
  $\trans_P$ under $\rho'$.
  Just like we did for the case of positive premises, we can conclude
  that there is some $\vec{z'}\in\mathcal{Z}$ such that
  $\rho(t_n(\vec{z'})) \relR_P \rho'(t_n(\vec{z}))$.
  Moreover, $t_n(\vec{z'})\ntrans[b_n]$ is also a premise in $r$ and
  hence
  ${\trans_P} \models \rho(t_n(\vec{z'})\ntrans[b_n])$,
  by ($\rho$\ref{rho:iv}).
  By the definition of $\relR_P$, two cases arise:
  \begin{enumerate}
  \item%
    $\rho(t_m(\vec{z'})) \bisim \rho'(t_m(\vec{z}))$,
    and hence ${\trans_P}\models\rho'(t_n(\vec{z}))\ntrans[b_n]$.
  \item%
    There is a function name $g\in F$ and terms
    $u_k,u'_k\in\closedTerms$, $1\leq k\leq\rank(g)$, such that
    \begin{itemize}
    \item%
      $\rho(t_n(\vec{z'})) = g(u_1,\cdots,u_{\rank(g)})$,
    \item%
      $\rho'(t_n(\vec{z})) = g(u'_1,\cdots,u'_{\rank(g)})$, and
    \item%
      $u_k \relR_P u'_k$ for all $1\leq k\leq\rank(g)$.
    \end{itemize}
    Suppose towards a contradiction that there exists
    $\pi\in\Delta(\closedTerms)$ such that
    $\rho'(t_n(\vec{z}))\trans[b_n]\pi \in {\trans_P}$.
    Since 
    $$S(\rho(t_n(\vec{z'})),b_n)+S(\rho'(t_n(\vec{z})),b_n) <
    S(f(v_1,\cdots,v_{\rank(f)}),a)+S(f(w_1,\cdots,w_{\rank(f)}),a),$$
    applying the first induction hypothesis, we have that
    $\rho(t_n(\vec{z'}))\trans[b_n]\pi' \in {\trans_P}$ for some
    $\pi'\in\closedDTerms$. This contradicts the fact that
    ${\trans_P} \models \rho(t_n(\vec{z'})\ntrans[b_n])$.
  \end{enumerate}
  Therefore, we have that all negative premises of $r$ hold under
  substitution $\rho'$. That is, for all $n\in N$ and
  $\vec{z}\in\mathcal{Z}$,
  \begin{enumerate}[($\rho'$1)~]
    \setcounter{enumi}{3}
  \item\label{rhoprime:iv} %
    ${\trans_P}\models\rho'(t_n(\vec{z}))\ntrans[b_n]$.
  \end{enumerate}

  It remains to show that all quantitative premises of $r$ also hold
  under $\rho'$.
  So, let $\theta_l(Y_l)\gtgeq_{l,k} p_{l,k}$ be a quantitative
  premise in $r$.
  If $\zeta \in \Var(\theta_l)$, then 
  by condition~\ref{item:conditions_on_variables_of_theta} of 
  Def.~\ref{def:ntmufxnu} and since the rule has not free variables,
  $\zeta$ is a distribution variable.
  Moreover $\zeta$ appears in the target of a positive premise;
  then $\zeta = \mu^{\vec{z}}_m$ with $m \in M$ and $\vec{z} \in \mathcal{Z}$. 
  By condition ($\rho'$\ref{rhoprime:targetCond}), 
  $\rho'(\mu_m^{\vec{z}}) \relR_P \rho(\mu_m^{\vec{z}})$.
  By Lemma\ref{lemma:sust_de_Dvar} we have 
  $\rho'(\theta_l) \relR_P \rho(\theta_l)$.
  In addition,  by ($\rho$\ref{rho:v}), $\rho(\theta_l(Y_l)) \gtgeq_{l,k} p_{l,k}$.
  Then, we calculate:
  \begin{align}
    \rho'(\theta_l)(\rho'(Y_l))
    & = \sum_{d'\in{\rho'(Y_l)/{\relR_P}}} \rho'(\theta_l)(d')
    \nonumber \\
    & = \sum_{d\in{\rho(Y_l)/{\relR_P}}} \rho'(\theta_l)(d^\uparrow)
    \tag{by ($\rho'$\ref{rhoprime:iii:d})} \\
    & = \sum_{d\in{\rho(Y_l)/{\relR_P}}} \rho(\theta_l)(d^\uparrow)
    \tag{by Lemma~\ref{lemma:sust_de_Dvar}} \\
    & \geq \sum_{d\in{\rho(Y_l)/{\relR_P}}} \rho(\theta_l)(d)
    \tag{$d\subseteq d^\uparrow$} \\
    & = \rho(\theta_l)(\rho(Y_l))
    \nonumber \\
    & \gtgeq_{l,k} p_{l,k} 
    \tag{since $\rho(\theta_l(Y_l))\gtgeq_l p_l$ holds}
  \end{align}
  So all quantitative premises of $r$ hold under $\rho'$, that is, for
  all $l\in L$, $k \in K_l$ and $\vec{z}\in\mathcal{Z}$,
  \begin{enumerate}[($\rho'$1)~]
    \setcounter{enumi}{4}
  \item\label{rhoprime:vii} %
    $\rho'(\theta_l (Y_l))\gtgeq_{l,k} p_{l,k}$.
  \end{enumerate}

  Since all premises of $\rho'(r)$ hold, we may conclude that
  $\rho'( f(x_1,\ldots,x_{\rank(f)}) \trans[a] \theta) \in {\trans_P}$.

  Finally, it remains to prove that
  $\rho(\theta) \relR_P \rho'(\theta)$. 
  By Lemma~\ref{lemma:sust_de_Dvar}, it is straightforward.
  Recall that the rule has not free variable,  
  then if $\zeta \in \Var(\theta)$ and $\zeta$ is a 
  a distribution varible we have 
  $\zeta = \mu^{\vec{z}}_m$ for some $\vec{z}\in\mathcal{Z}$ and
  $m\in M$.
  By condition ($\rho'$\ref{rhoprime:targetCond}) 
  $\rho(\mu_m^{\vec{z}}) \relR_P \rho'(\mu_m^{\vec{z}})$.
  On the other hand, if $\zeta$ is a term variable,
  by conditions ($\rho'$\ref{rhoprime:i}) and  
  ($\rho'$\ref{rhoprime:iii:a}) it holds $\rho(\zeta) \relR_P \rho'(\zeta)$,
  i.e. the condition of Lemma~\ref{lemma:sust_de_Dvar} are satisfied.
\end{proof}

\section{Proof of Lemma~\ref{lemma:closure}} \label{ap:lemma:closure}

\begin{proof}[Proof of Lemma~\ref{lemma:closure}]
 We prove first the left to right implication. 
 Suppose $\dedrule{H}{c}$ is provable from $P$. 
 We proceer by induction on the partial well-order $<$ on proof structures. 
 If $c \in H$, then  $\dedrule{H}{c}\in R^\proves$.
 For the induction case, let $\dedrule{H}{\sigma(\conc{r})}$ be provable from $P$ by means of a proof structure 
 $(B, r, \phi)$ and a matching substitution $\sigma$. 
 Assume that any term provable by means of a smaller proof structure belongs to $R^\proves$,
 except for the closed quantitative literals. 
 Then $\dedrule{H}{\sigma(p)}$ for $p \in \pprem{r} \cup \qprem{r}$ and for 
 $p \in \qprem{r}$ with $\sigma(p)$ an open term. 
 In addition, by Def.~\ref{def:provable}, if $p \in \qprem{r}$ with $\sigma(p)$  being closed, $\sigma(p)$ holds. Therefore $\dedrule{H}{\sigma(\conc{r})} \in R^\proves$.

 For the right to left implication, we apply induction on the construction of $R^\proves$.
 The base case is trivial.
  Suppose $r \in R$ is such that there is a substitution $\rho$ such that 
  for all premise $p \in \prem{r}$, $\dedrule{H}{\rho(p)}$ holds or
  $\rho(p)$ is a valid closed quantitative premise.
   Let $(B_p, r_p, \phi_p)$ be the proof structure with matching substitution $\sigma_p$ 
   associated to $\dedrule{H}{\rho(p)}$ with $p \in \prem{r}$ and $p$ is not a closed quantitative premise. 
   Since there exist at least as many variables as there are premises in K, the variables in these
   proof structures can be renamed to become all different, and different from the ones in $r$, and a
   substitution $\sigma$ can be constructed that matches with each of these proof structures so as to yield the
   corresponding rule, and equals $\rho$ on the variables in $r$. 
   Then, $(\cup_{p \in \prem{r}^*} B_p \cup \{r\}, r, \cup_{p \in \prem{r}^*} \phi_p \cup \phi')$
   with $\prem{r}^*$ the set of premises of $r$ that are not closed quantitative literals  
   and $\phi'$ the function that maps $r_p$ to $p$ for all $p \in \prem{r}^*$  is a proof structure 
   that matches with $\sigma$, then $\dedrule{H}{\sigma(\conc{r})}$
\end{proof}

\section{Proof of Lemma~\ref{lemma:WSPsubsumesStra}} \label{ap:WSPsubsumesStra}
 The proof uses the machinery introduced in Section~\ref{sec:proofStructure} 
 and the following Lemma.

\begin{lemma}\label{lemma:aux}
 Let $P$ be a PTSS and let $\dedrule{H}{c}$ be provable from $P$.
 If $\trans_{P} \models \psi$ for all $\psi \in H$ then $\trans_{P} \models c$.
\end{lemma}

The proof of the last lemma is straightforward by induction in the proof structure.
We proceed with the proof of Lemma~\ref{lemma:WSPsubsumesStra}.

\begin{proof}[Proof of Lemma~\ref{lemma:WSPsubsumesStra}.]
($\Leftarrow$) We proceed by transfinite induction.
Suppose  ${{\trans_{P_{\beta'}}} \models \psi} \implies P \proves_{ws} \psi$ 
for all $\beta'< \beta$ and  ${\trans_{P_{\beta, j'}}} \models \psi \implies P \proves_{ws} \psi$ 
for all $j'<j$.
If we prove ${\trans_{P_{\beta, j}}} \models \psi \implies P \proves_{ws} \psi$
then ${\trans_{P}} \models \psi \implies P \proves_{ws} \psi$.
 
Let $\psi$ be a positive literal such that $\psi \in {\trans_{P_{\beta, j}}}$,
then there is a rule $r \in R$ and a substitution $\rho$ such that
$\rho(\conc{r}) = \psi$  and $\bigcup_{\beta'<\beta}\trans_{P_{\beta'}} \cup
\bigcup_{j'< j}\trans_{P_{\beta,j'}} \models \prem{r}$.
By induction hypothesis there are derivation for all premises in $\prem{r}$.
Then, using $r$ and $\rho$ we can construct a new proof for $\psi$,
hence  $P \proves_{ws} \psi$.

Let $\psi$ be a negative literal $t \ntrans[a]$ such that ${\trans_{P}} \models \psi$. 
Then for all  $\pi \in \Delta(\closedTerms)$,  $S(t \trans[a] \pi) < \beta$.
Moreover, for each $t \trans[a] \pi$, ${\trans_{P}} \not \models t \trans[a] \pi$.
Therefore for each set of closed literals $H$, 
such that $P \proves \dedrule{H}{t \trans[a] \pi}$,
(here $\proves$ refers to the Def~\ref{def:provable},
then $H$ does not contain quantitative literals because the premises are closed),
there is a $\delta \in H$ such that $\bigcup_{\beta' < \beta}{\trans_{\beta'}} \not \models \delta$,
otherwise, $t \trans[a] \pi$ would be derivable by Lemma~\ref{lemma:aux}.
Notice that by the definition of stratification (Def.~\ref{def:stratification}) 
the required premises to derive $t \trans[a] \pi$  belong to a stratum less than or equal to $\beta$.
This implies that there is a literal $\delta'$, which denies $\delta$, such that
$\bigcup_{\beta' < \beta} {\trans_{\beta'}} \models \delta'$. 
In particular, this holds for $H$ with only negative premises and 
in view of $\bigcup_{\beta' < \beta} {\trans_{\beta'}} \models \delta'$,
by induction hypothesis,$P \proves_{ws} \delta'$.
Therefore, by Def.~{\ref{def:ws-proof}}, $P \proves_{ws} \psi$.

\medskip
($\Rightarrow$)
Suppose $P \proves_{ws} \psi$ and $\psi$ is negative. 
Because $\proves_{ws}$ is consistent, Lemma~\ref{lemma:completeThenCons},
for all premise $\psi'$ that denies $\psi$,
$P \not \proves_{ws} \psi'$. 
By ``$\Leftarrow$ case'',
and its contraposition,
$\trans_{P} \not \models \psi'$, and therefore
$\trans_{P} \models \psi$.
Applying a similar reasoning, suppose $P \proves_{ws} t \trans[a] t'$.
Then $P \not \proves_{ws} t \ntrans[a] t'$ and again 
$\trans_{P} \not \models t \ntrans[a] t'$. 
Finally $\trans_{P} \models t \trans[a] t'$. 
\end{proof}

\section{Proof of Lemma~\ref{lemma:sameNegRules}}\label{ap:lemma:sameNegRules}

\begin{proof}
Suppose
$P \proves \dedrule{N}{\phi} \Leftrightarrow P' \proves \dedrule{N}{\phi}$ holds.
We will only prove $P \proves_{ws}{\psi} \implies P' \proves_{ws} {\psi}$ 
for all closed (positive/negative) literal $\psi$.
The other case is symmetric.
We proceed by induction on the size of the derivation tree. 
The base case is straightforward.

If $\psi$ is a negative premise, by Def.~\ref{def:ws-proof}, 
all premises $\psi_k$ over $\psi$ are positive and each one denies a literal in a set $N$, 
where $N$ is such that $P \proves \dedrule{N}{\psi'}$, where $\psi'$ denies $\psi$.
By the hypothesis of the lemma $P \proves \dedrule{N}{\psi'} \Leftrightarrow P' \proves \dedrule{N}{\psi'}$, 
and, by induction hypothesis, $P \proves_{ws}{\psi_k} \implies P' \proves_{ws} {\psi_k}$,
therefore $P' \proves_{ws} {\psi}$.

Let $\psi$ be a positive premises.
Let $p_{\psi}$ a well supported proof of $\psi$; 
besides, let $N_\psi$ be the negative literals that appear in $p_{\psi}$ and 
$p'_{\psi}$ be the  well supported proof obtained by removing from $p_{\psi}$
the literals in $N_\psi$ with the sub well supported structure used to prove it.
Then $P \proves \dedrule{N_\psi}{\psi}$, because the proof structure and the substitution 
can be defined using $p'_\psi$. 
By hypothesis $P' \proves \dedrule{N_\psi}{\psi}$.
Repeating the reasoning of the previous case 
$P \proves_{ws} \psi'$ for all $\psi' \in N_\psi$ then 
$P' \proves_{ws} \psi'$ for all $\psi' \in N_\psi$.
Therefore $P' \proves_{ws} \psi$.
\end{proof}

\end{document}

%% file: macros.tex
\settimeformat{hhmmsstime}
\newcommand{\builddate}{\ifnumcomp{\year}{<}{10}{0}{}\the\year-\ifnumcomp{\month}{<}{10}{0}{}\the\month-\ifnumcomp{\day}{<}{10}{0}{}\the\day\ \currenttime}
\ifnumcomp{\svnhour}{=}{23}{\def\svnhour{01}}{}
\ifnumcomp{\svnhour}{=}{22}{\def\svnhour{00}}{}
\ifnumcomp{\svnhour}{=}{21}{\def\svnhour{23}}{}
\ifnumcomp{\svnhour}{=}{20}{\def\svnhour{22}}{}
\ifnumcomp{\svnhour}{=}{19}{\def\svnhour{21}}{}
\ifnumcomp{\svnhour}{=}{18}{\def\svnhour{20}}{}
\ifnumcomp{\svnhour}{=}{17}{\def\svnhour{19}}{}
\ifnumcomp{\svnhour}{=}{16}{\def\svnhour{18}}{}
\ifnumcomp{\svnhour}{=}{15}{\def\svnhour{17}}{}
\ifnumcomp{\svnhour}{=}{14}{\def\svnhour{16}}{}
\ifnumcomp{\svnhour}{=}{13}{\def\svnhour{15}}{}
\ifnumcomp{\svnhour}{=}{12}{\def\svnhour{14}}{}
\ifnumcomp{\svnhour}{=}{11}{\def\svnhour{13}}{}
\ifnumcomp{\svnhour}{=}{10}{\def\svnhour{12}}{}
\ifnumcomp{\svnhour}{=}{09}{\def\svnhour{11}}{}
\ifnumcomp{\svnhour}{=}{08}{\def\svnhour{10}}{}
\ifnumcomp{\svnhour}{=}{07}{\def\svnhour{09}}{}
\ifnumcomp{\svnhour}{=}{06}{\def\svnhour{08}}{}
\ifnumcomp{\svnhour}{=}{05}{\def\svnhour{07}}{}
\ifnumcomp{\svnhour}{=}{04}{\def\svnhour{06}}{}
\ifnumcomp{\svnhour}{=}{03}{\def\svnhour{05}}{}
\ifnumcomp{\svnhour}{=}{02}{\def\svnhour{04}}{}
\ifnumcomp{\svnhour}{=}{01}{\def\svnhour{03}}{}
\ifnumcomp{\svnhour}{=}{00}{\def\svnhour{02}}{}
\newcommand{\wipinfo}{
  \ \ Last change: \svnyear-\svnmonth-\svnday\ \svnhour:\svnminute:\svnsecond \\
  \qquad\quad Build: \builddate \\
  \ \qquad\, Revision: \svnfilerev \qquad\qquad\qquad\qquad\ \ 
}

\newcommand{\N}{\mathbb{N}} 
\newcommand{\Q}{\mathbb{Q}} 

\newcommand{\cardof}[1]{|{#1}|}

\newcommand{\ddedrule}[2]{\frac{\displaystyle #1}{\displaystyle #2}}
\newcommand{\dedrule}[2]{\frac{#1}{#2}}
\newcommand{\trans}[1][]{\xrightarrow{\, {#1} \, }}
\newcommand{\ntrans}[1][]{\mathrel{{\trans[#1]}\makebox[0em][r]{$\not$\hspace{2ex}}}{\!}}
\newcommand{\gtgeq}{\trianglerighteq}

\newcommand{\strans}[1][]{\stackrel{#1}{\longrightarrow}}

\newcommand{\tuple}[1]{\langle{#1}\rangle}

\newcommand{\rank}{\mathop{\sf r}}
\newcommand{\openT}{\mathbb{T}}
\newcommand{\openTerms}{\openT(\Sigma)}
\newcommand{\closedTerms}{T(\Sigma)}
\newcommand{\openDT}{\mathbb{DT}}
\newcommand{\openDTerms}{\openDT(\Sigma)}
\newcommand{\closedDTerms}{\DT(\Sigma)}
\newcommand{\Var}{\mathop{\textit{Var}}}
\newcommand{\PTrn}{\textit{PTr}}
\newcommand{\PTr}{\PTrn(\Sigma, A)}

\newcommand{\pprem}[1]{\textrm{pprem}(#1)}
\newcommand{\nprem}[1]{\textrm{nprem}(#1)}
\newcommand{\qprem}[1]{\textrm{qprem}(#1)}
\newcommand{\prem}[1]{\textrm{prem}(#1)}
\newcommand{\conc}[1]{\textrm{conc}(#1)}

\newcommand{\Red}[2]{\textrm{Red}}
\newcommand{\nvdg}{n_{\mathrm{VDG}}}

\newcommand{\TVar}{\mathcal{V}}
\newcommand{\PVar}{\mathcal{M}}
\newcommand{\DVar}{\mathcal{D}}

\newcommand{\support}{\mathsf{support}}

\newcommand{\V}{\mathcal{V}}
\newcommand{\M}{\mathcal{M}}

\newcommand{\degPTSS}{\textsf{D}}

\newcommand{\Act}{A}

\newcommand{\ntmufnu}{\ensuremath{\textit{nt}\mu\textit{f}\nu}}
\newcommand{\ntmuxnu}{\ensuremath{\textit{nt}\mu\textit{x}\nu}}
\newcommand{\ntmufxnu}{\ensuremath{\ntmufnu\textit{/}\ntmuxnu}}

\newcommand{\ntmuft}{\ensuremath{\textit{nt}\mu\textit{f}\theta}}
\newcommand{\ntmuxt}{\ensuremath{\textit{nt}\mu\textit{x}\theta}}
\newcommand{\ntmufxt}{\ensuremath{\ntmuft\textit{/}\ntmuxt}}
\newcommand{\bntmuft}{\ensuremath{\textit{\bfseries nt}\boldmath\mu\textit{\bfseries f}\boldmath\theta}}
\newcommand{\bntmuxt}{\ensuremath{\textit{\bfseries nt}\boldmath\mu\textit{\bfseries x}\boldmath\theta}}
\newcommand{\bntmufxt}{\ensuremath{\bntmuft/\bntmuxt}}

\newcommand{\nxmuft}{\ensuremath{\textit{nx}\mu\textit{f}\theta}}

\newcommand{\nxmutt}{\ensuremath{\textit{nx}\mu\textit{t}\theta}}

\newcommand{\ntyft}{\ensuremath{\textit{ntyft}}}
\newcommand{\ntyxt}{\ensuremath{\textit{ntyxt}}}
\newcommand{\ntyfxt}{\ensuremath{\ntyft\textit{/}\ntyxt}}

\newcommand{\Diag}{\mathop{\textsf{Diag}}}
\newcommand{\diag}[1]{\textsf{Diag}\{#1\}}

\newcommand{\partZ}{\Xi}

\newcommand{\limp}{\Rightarrow}

\definecolor{lightblue}{RGB}{224,224,255}
\definecolor{lightred}{RGB}{255,224,224}
\definecolor{lightgreen}{RGB}{224,255,224}
\definecolor{lightyellow}{RGB}{255,255,224}
\definecolor{lightpurple}{RGB}{255,224,255}
\definecolor{darkerred}{RGB}{64,0,0}
\definecolor{darkred}{RGB}{128,0,0}
\definecolor{darkblue}{RGB}{0,0,128}
\definecolor{darkgreen}{RGB}{0,128,0}
\definecolor{darkpurple}{RGB}{128,0,128}

\newcommand{\colorpar}[3]{\colorbox{#1}{\parbox{#2}{#3}}}

\newcommand{\marginremark}[3]{\marginpar{\colorpar{#2}{\linewidth}{\color{#1}#3}}}

\makeatletter
\def\THICKhrulefill{\leavevmode \leaders \hrule height 5pt\hfill \kern \z@}
\makeatother

\newcommand{\remarkPRD}[1]{\marginremark{darkred}{lightred}{\tiny{[PRD]~ #1}}}
\newcommand{\remarkMDL}[1]{\marginremark{darkpurple}{lightpurple}{\tiny{[MDL]~ #1}}}
\newcommand{\remarkDG}[1]{\marginremark{darkgreen}{lightyellow}{\tiny{[DG]~ #1}}}

\newcommand{\bisim}{\sim}

\newcommand{\relR}{\mathrel{\textsf{R}}}

\newcommand{\closed}[1]{{#1}\text{-closed}}

\newcommand{\liff}{\Leftrightarrow}
\renewcommand{\implies}{\Rightarrow}

\newcommand{\proves}{\vdash}
\newcommand{\rtop}{\text{top}}
\newcommand{\qtop}{\text{qtop}}

\newcommand{\DT}{\textsf{DT}}

\newcommand{\wsproves}{\vdash_{\textit{ws}}}